\PassOptionsToPackage{prologue,dvipsnames,svgnames}{xcolor}
\documentclass[acmsmall,screen,natbib=false,nonacm]{acmart}

\setcopyright{acmlicensed}
\copyrightyear{2018}
\acmYear{2018}
\acmDOI{XXXXXXX.XXXXXXX}

\usepackage[
  datamodel=acmdatamodel,
  style=acmauthoryear,
]{biblatex}

\usepackage{stmaryrd}
\usepackage{listings}
\usepackage{amsmath}
\usepackage{tikz}
\usepackage{tikz-cd}
\tikzset{>={Straight Barb[scale=0.8]}}
\tikzcdset{
  arrow style=tikz
}
\usetikzlibrary{positioning, calc}

\usepackage[inline]{enumitem}

\usepackage{csquotes}

\usepackage{macros}

\usepackage[capitalize]{cleveref}
\begin{document}

%%
%% The "title" command has an optional parameter,
%% allowing the author to define a "short title" to be used in page headers.
\title{Definitional Inversion, Without Normalisation}

%%
%% The "author" command and its associated commands are used to define
%% the authors and their affiliations.
%% Of note is the shared affiliation of the first two authors, and the
%% "authornote" and "authornotemark" commands
%% used to denote shared contribution to the research.

\author{Mario Carneiro}
\email{marioc@chalmers.se}
\orcid{0000-0002-0470-5249}

\author{Thierry Coquand}
\email{Thierry.Coquand@cse.gu.se}
\orcid{0000-0002-5429-5153}
\affiliation{%
  \institution{Chalmers University of Technology}
  \city{Gothenburg}
  \country{Sweden}
}

\author{Adrien Frabetti Mathieu}
\email{adrien.mathieu@ens.psl.eu}
\orcid{0009-0006-1590-8777}
\affiliation{%
  \institution{École Normale Supérieure PSL, IRIF}
  \city{Paris}
  \country{France}
}

\author{Meven Lennon-Bertrand}
\email{meven.bertrand@inria.fr}
\orcid{0000-0002-7079-8826}
\affiliation{%
  \institution{Université Paris Cité, INRIA, CNRS, IRIF}
  \city{Paris}
  \country{France}
}

\author{Paul-André Melliès}
\email{mellies@irif.fr}
\orcid{0000-0001-6180-2275}
\affiliation{%
  \institution{Université Paris Cité, INRIA, CNRS, IRIF}
  \city{Paris}
  \country{France}
}

\author{Stephanie Weirich}
\email{sweirich@seas.upenn.edu}
\orcid{0000-0002-6756-9168}
\affiliation{%
  \institution{University of Pennsylvania}
  \city{Philadelphia}
  \state{Pennsylvania}
  \country{United States}
}

%%
%% By default, the full list of authors will be used in the page
%% headers. Often, this list is too long, and will overlap
%% other information printed in the page headers. This command allows
%% the author to define a more concise list
%% of authors' names for this purpose.
\renewcommand{\shortauthors}{Carneiro, Coquand, Frabetti Mathieu, Lennon-Bertrand, Melliès and
Weirich}

%%
%% The abstract is a short summary of the work to be presented in the
%% article.
\begin{abstract}
  We contribute a new proof technique, based on domain theory, to prove key meta-theoretic
  properties of dependent type systems: definitional inversion properties, \ie injectivity and
  no-confusion of type constructors. This proof technique is independent of normalisation, and
  indeed applies even for the “type-in-type” rule of Martin-Löf's original
  type theory. Our proof is the first to establish injectivity of type constructors for such a
  system in the presence of η laws.
  More generally, the technique is motivated by, and intended for, the metatheory of
  systems such as Idris, Lean, or dependent Haskell, whose underlying type
  theory is known to be non-normalising, as well as projects such as MetaRocq
  or Lean4Lean, where Gödel's second incompleteness theorem means we cannot
  show normalisation of the object logic in itself.
  We showcase the method on a small type theory, then explain how it extends to
  more ambitious extensions.
\end{abstract}

%%
%% The code below is generated by the tool at http://dl.acm.org/ccs.cfm.
%% Please copy and paste the code instead of the example below.
%%

% \ccsdesc[500]{Do Not Use This Code~Generate the Correct Terms for Your Paper}
% \ccsdesc[300]{Do Not Use This Code~Generate the Correct Terms for Your Paper}
% \ccsdesc{Do Not Use This Code~Generate the Correct Terms for Your Paper}
% \ccsdesc[100]{Do Not Use This Code~Generate the Correct Terms for Your Paper}

%%
%% Keywords. The author(s) should pick words that accurately describe
%% the work being presented. Separate the keywords with commas.

% \received{20 February 2007}
% \received[revised]{12 March 2009}
% \received[accepted]{5 June 2009}

%%
%% This command processes the author and affiliation and title
%% information and builds the first part of the formatted document.
\maketitle

\section{Introduction}
\label{sec:intro}

Dependent type systems are the foundation of a growing ecosystem of
proof assistants (\Agda, \Lean, \Rocq{}…) and programming languages
(\Idris, dependent \Haskell{}…), themselves underlying many critical
software systems and libraries of mathematics. Given their importance,
we need a good understanding of these formalisms, which goes through
\emph{meta-theory}: the study of these systems and mathematical proof of their
properties. Sadly, the current toolbox of meta-theory is lacking: the main
two kinds of approaches, confluence and logical relations,
are severely limited. Confluence cannot deal with important features of modern
dependent type systems, particularly η laws. Traditional logical relations,
because they rely on normalisation, do not apply to non-normalising type
theories, which are plenty in practice (\Idris, dependent \Haskell{}, \Lean{}…).
Also, since normalisation implies logical consistency, one needs to work
in a very powerful ambient logic, which must be in particular stronger than
the object logic. For projects such as \agdac{} \cite{Liesnikov2025},
\leanf{} \cite{Carneiro2024}, or \MetaRocq{} \cite{MetaCoq2025}, which aim
at studying the meta-theory of major proof assistants in themselves, this
is a non-starter.

We improve this situation by proposing a new technique, based on domain
theory, which lifts these limitations: our technique works in a weak ambient
logic, and applies even to non-normalising type theories; yet, it easily
incorporates η laws. By design, this technique does not prove normalisation,
however, it is able to prove definitional inversion principles, such as the
celebrated injectivity of type constructors, which in many aspects is just as
critical, if not more, than normalisation. Consequences are plenty: progress
and preservation \cite{Wright1994}, thus soundness of the language viewed as
programming language \cite{DependentHaskell2017,MetaCoq2025},
uniqueness of types \cite{MetaCoq2025,Carneiro2024}, kernel verification
\cite{LennonBertrand2025}…

\paragraph{Definitional inversions}
What sets dependent type systems
distinctively apart from other programming languages or logics
is \emph{definitional equality} (or \emph{conversion}), written
\(\convop\), a rich equational
theory that is directly incorporated in the type system.
Given its centrality, it is no surprise that the most important meta-theoretic
properties of dependent types involve definitional equality.
A very natural thing to ask is what we can deduce from equations between types.
For simple types it is evident that
\(\nat \nconvop A \to B\), or that \(A \to B \convop A' \to B'\) implies that
\(A \convop A'\) and \(B \convop B'\): type constructors are \emph{disjoint}
(what is also called \emph{no-confusion}) and
\emph{injective}. For dependent type theory, by contrast, this is absolutely
not evident: one might be able to prove \(\nat \convop \nat \to \nat\) through
a long chain of equations which brings in the entire equational theory of the
term language. In extensional type theory (ETT), for instance, this equation
holds whenever the context is inconsistent.

Yet, these properties -- which we collectively refer to as
\emph{definitional inversion principles} -- are extremely useful to establish
other good properties of the type system. Intuitively, this is because
definitional inversions imply that, despite the language's equations, matching a
type against a pattern, \eg \(\_ \to \_\), is a valid operation: by
no-confusion, a type cannot match non-overlapping patterns (\eg \(\_ \to \_\)
and \(\nat\)), and by injectivity the result of a match is well-defined
(\eg we can talk without ambiguity of “the domain” of a function type). This is
a key ingredient in many meta-theoretic proofs, including the above.
In contrast, ETT, which lacks definitional inversions, has a poorly-behaved
operational semantics and does not have unique types. Moreover,
the design of semi-decision procedures for typing in ETT which do better than naïve
enumeration of derivations is an open problem \cite{PetkovicKomel2021}.

\paragraph{Definitional inversion via confluence}

The problem of definitional inversion is as old as dependent type theory itself.
\Posscite{MartinLoef1971} technical report, which can be considered
the birth paper of dependent type theory, deals with it,
since definitional inversion is crucial for proving subject reduction,
a required property of any reasonable type system.
In order to prove subject reduction, the report's key lemma is
definitional inversion, in the form of injectivity of the dependent product
operation. For this lemma, Martin-Löf proved confluence
(the Church-Rosser property) adapting \posscite{Tait1967}
method of parallel reduction.%
\footnote{This proof was historically important: it was reproduced
  in \posscite{Barendregt1971} thesis, which was one motivation for
  \posscite{Plotkin2004} formulation of operational semantics.  Also, through the work on
  the Calculus of Construction \cite{Coquand1988}, this has been the main tool
  to justify subject reduction for pure type systems \cite{Barendregt1992},
all to the current mechanised meta-theory of \Rocq \cite{MetaCoq2025}.}

In 1971, Martin-Löf did not define typed definitional equality as a judgement: conversion was
β-conversion, defined at the level of raw terms.  Such a presentation, however, makes it is
difficult to give a denotational semantics to the language. Later papers
\cite{MartinLoef1984,Cardelli1986,MeyerReinhold1986} hence present conversion as a (typed)
judgement instead. The problem with this approach is that typing and conversion are mutually
defined, and confluence cannot be used to prove the key injectivity lemma \footnote{This is why
  \eg the report \cite{Reinhold1989Typechecking}, which proves some results announced in
\textcite{MeyerReinhold1986}, instead uses the presentation with conversion on raw terms.}

For a type system with only β-conversion, it is possible to show equivalence between the typed
and untyped presentations purely syntactically \cite{Adams2006,Siles2012}. With η-conversion,
however, the presentation as typed judgement is clear, but its equivalence with an untyped
presentation is unclear, or even what this untyped presentation should be!  The reason is that
η plays poorly with confluence. Counter-examples have been known for long. A first, due to
Nederpelt (see \eg \textcite{GeuversWerner1994}) considers the raw term
$\lambda_{x:A}(\lambda_{y:B} y)~x$, which reduces both to $\lambda_{y:B}y$ by η-reduction and
to $\lambda_{x:A}x$ by β-reduction. Reduction cannot be confluent if $A$ and $B$ are not
convertible. Worse, Klop showed in his PhD thesis \cite[Chap. III]{Klop1980} that η-reduction
for pairs \(t \convop \pair{\fst{t}}{\snd{t}}\), is not confluent, even barring type
annotations.  There is one presentation of Pure Type Systems with η in
\textcite{GeuversWerner1994}, but it is proved in \textcite{Barthe2006} that the system is not
confluent.  In fact, we do not think that this presentation is equivalent to the more
conceptual presentation as judgement that we use in the present paper.%
\footnote{For a similar reason, the question of existence of a fixed point combinator, which is
  solved positively in \textcite{Barthe2006}, seems actually to be \emph{open} for the version
we present here.}

Recent works have managed to negotiate these issues slightly:
\textcite{FelicissimoWinterhalter2026} handle a very limited form of η,
and \textcite{Liu2026} manage to prove confluence for untyped
reduction in a calculus with η for pairs, but only for strongly normalising terms.
But, by and large, mixing confluence on the one hand and η laws and typed conversion
remains difficult.

\paragraph{Definitional inversion via normalisation}
Yet, for both theoretical and practical reasons,
having η-conversion is actually crucial for developing mathematics in dependent
type systems. To name only two examples, η for records is needed to deal with
hierarchies of structures in \Lean \cite{Wieser2023}, and η for functions is
used to derive function extensionality from univalence \cite[Sec.
4.9]{UniFoundationsProgram2013}. They are thus included in \Agda, \Lean and \Rocq.
Even the venerable \textsf{Automath}
\cite{deBruijn1970Automath} insisted from the start on having η-conversion.

To deal with rich equational theories including η conversion%
\footnote{And others, for instance definitional functoriality laws \cite{Subtyping2024}.}
and with definitional equality as a typed judgement,
one can reach for a radically different, much more semantic approach to meta-theory,
via normalisation proofs and logical relations.
A large array of these techniques exist, from relatively elementary
versions \cite{HarperPfenning2005,Abel2013} to more abstract
ones \cite{Sterling2021a,Bocquet2023}.
The shape of such an argument goes back to \textcite{Gandy1956}: to show the axiom
of extensionality admissible, he defines a relation by recursion on the type,
lifts it pointwise to environments, and proves that every well-typed term maps
related environments to related values --~that is, that every definable term is
extensional.

While powerful, by construction they only apply to
normalising systems, and need an ambient theory which is logically stronger than
the object one, since normalisation implies logical consistency.
Such proofs typically rely on strong logical principles
such as induction-recursion \cite{Coquand1991Algorithm,Abel2017},
or impredicativity \cite{Gratzer2019,Jang2025}, although with
enough care a few extra universes are sufficient \cite{Adjedj2024}.
This situation is not satisfactory: one would expect definitional inversions,
which are purely syntactic properties to be provable in a weak ambient theory.
This is particularly critical for projects aiming at studying the meta-theory
of proof assistants in themselves: \agdac{} \cite{Liesnikov2025},
\leanf{} \cite{Carneiro2024}, or \MetaRocq{} \cite{MetaCoq2025}.
Similarly, non-normalising type theories are out of scope of these techniques,
which rules out many dependently-typed programming languages
\cite{Cardelli1986,Augustsson1998,DependentHaskell2017,IdrisTypeType},
which, for convenience and simplicity, use type-in-type, but also
other interesting non-normalising features, such as proof-irrelevant accessibility
as used in \Lean \cite{Felicissimo2026}.

In fact, we do not need to reach for such complicated type theories:
an extremely simple system that already crystallises the issues at hand is \MLTTo,
a variation on Martin-Löf's original
type theory \cite{MartinLoef1971} with a type \(\P\) of dependent functions,
β and η equations, and a universe with \(\univ \ty \univ\).%
\footnote{This is not exactly Martin-Löf's 1971 system, which lacked
η, and could thus be treated using confluence.}
The extensionality equation defeats confluence-based approaches,
while type-in-type makes the system non-normalising \cite{Girard1972},
defeating normalisation-based proofs. Despite its simplicity and importance,
little was known of \MLTTo{} before the present work: for instance,
subject reduction was open.

\paragraph{Definitional inversion via domain theory}
The present work addresses this blind spot: we develop a new proof technique
that is able to simultaneously cope with η laws, without the limitations of
traditional logical relations. The key technical ingredient is to use domain theory,
based on ideas first explored by \textcite{Coquand2018a}.
We first demonstrate the technique on \MLTTo{}, contributing the first
proof of definitional inversion principles for it.%
\footnote{This in particular solves the open problem raised at the end of
\textcite{Coquand2018a}.}
We then show that it
robustly extends to richer type theories. Indeed, we expect our technique
to extend to very complex type theories, all the way to \leanf and \MetaRocq.

The originality of our approach is to use \emph{domain theory}
as a way to avoid the strong logical principles necessary for traditional
logical relation proofs. We use a \emph{finitary projection} model, where a
type is interpreted by a finitary projection of some domain.
This idea goes back at least to \posscite{Scott1981} lecture on domain theory,
and it was studied systematically by \textcite{McCracken1982} for a semantics of
polymorphism.\footnote{Another
  semantics, where a type is interpreted as a functor on the category of domains with
  embedding projection pairs, was developed in \textcite{CoquandGunterWinskel1989}.
  This was inspired by \posscite{Girard1986SystemF} work, which itself motivated the creation
of linear logic.}  A semantics inspired by this finitary projection model can be found in
\textcite{AmadioBruceLongo1986,Cardelli1986}. The latter presents conversion as a judgement,
but does not discuss the problem of subject reduction for β-reduction.

The key aspect of domain theory we exploit is to provide a method
to derive \emph{operational} results about programs by \emph{denotational}
reasoning, via an \emph{adequacy theorem}.\footnote{See \eg \textcite{Plotkin1977};
  if, for instance, the semantics of a program is
$0$, adequacy tells us the program will terminate and compute $0$.}
Our main technical result is a refined version of adequacy, showing
that if two types have for semantics $\cuniv$, then they
both reduce (in a type-correct way) to $\univ$, or if the semantics is a product type,
then both reduce to a product with convertible components. Perhaps surprisingly,
we do not need strong logical principles, such as induction-recursion,
to show this: basic induction on natural numbers is enough.

\paragraph{Contributions}

We provide a new proof technique, based on domain theory, to show definitional inversion
principles in dependent type theories. This applies to non-normalising type theories,
and supports η-conversion.  In \Cref{sec:domain-model}, we explain the technique on a very
simple type theory, \MLTTo{}, with only function types and a universe
with “type-in-type”.
In \Cref{sec:extensions} we show how the technique scales
to richer type theories, covering many theoretical challenges of the type theory of
modern proof assistants such as \Lean or \Rocq. We treat:
\begin{itemize}
  \item dependent sums \(\Sigma\) and a unit type \(\unit\), both with their η laws
    (surjective pairing);
  \item  a fixed-point combinator \(\fix\), which breaks
    normalisation even without type-in-type;
  \item natural numbers with pattern-matching, which combined with \(\fix\)
    give recursion;
  \item a proof-relevant identity type \(\Id\) with transport \(\tr\);
  \item a universe of propositions \(\prop\) with definitional proof
    irrelevance.
    % \item intersection and singleton types, arising naturally from the domain
    %   structure;
    % \item \Lean-style \(\prop\)-valued accessibility,
    %   which breaks normalisation \cite{Felicissimo2026}.
\end{itemize}
We have mechanised the proofs (thrice!) in \Agda, \Lean and \Rocq, and comment on
formalisation-specific aspects in \Cref{sec:formalisation}. The exact
rundown on what has been mechanised is summarised in the following table,
where proof assistant names link to the (anonymised) formalisations.

\begingroup
\setlength{\tabcolsep}{14pt}
\renewcommand{\arraystretch}{1.5}
\begin{tabular}{r|cccccc}
  & \(\univ \ty \univ \) & \(\P\) & \(\Sigma, 1\) & \(\fix\) & \(\nat\) + match & \(\Id\) \\
  \hline
  \href{https://anonymous.4open.science/r/domain-semantics-4ED6/README.md}{\Agda} &
  \cmark & \cmark & & \cmark & \cmark & \cmark \\
  \href{https://anonymous.4open.science/r/domain-semantics-lean-BEE7/README.md}{\Lean} & \cmark & \cmark & \cmark & \cmark & \cmark & transport only \\
  \href{https://anonymous.4open.science/r/domain-semantics-rocq-B18B/README.md}{\Rocq} &
  \cmark & \cmark & & & \(0\) and \(\suc\) only &
\end{tabular}
\endgroup

\section{A Domain Model of Type Theory}
\label{sec:domain-model}

An intuition from domain theory is that a program can be represented by its observable
behaviour: given a chosen set of predicates, each representing a finite
observation, a term is described by the set of predicates
it satisfies. Although the set of observations is
infinite (for functions, each input-output pair yields different
observations), a program can nonetheless be reconstructed from approximations
by finite sets of finite observations, embodied by \emph{compact} elements of the domain.
% This naturally incorporates non-termination:
% a diverging program is one that does not yield any interesting observation.

This approach gives us a new way to construct logical relations for dependent type
theory. In the
standard logical relations, one defines a predicate on terms by recursion on the
structure of types.  With dependent types, this is not directly possible since types themselves
compute, which means we cannot directly follow their structure.
Traditionally, the recursion instead goes on the types' normal forms \cite{Abel2013}.%
\footnote{Naturally leading to definitions by
induction-recursion \cite{Dybjer2000}.}  But this only works if types have a
normal form! Domain theory gives us a different solution: to each observation on
a type, which provides a partial view of its shape, we associate a
logical predicate. The type's predicate itself is obtained by aggregating these
partial predicates. This harmoniously accommodates non-terminating types,
instead of ruling them out. Although this is not
enough to prove normalisation, it suffices to derive the definitional inversion
properties we seek.
More precisely, the plan of this section is as follows:
\begin{itemize}
  \item we build in \cref{sec:fix-solution} a domain \(\dom\) as a solution to a
    domain equation;
  \item we introduce in \cref{sec:semantic-typing} a semantic notion of typing
    between elements of \(\dom\), which we reformulate in
    \cref{sec:finitary-projectors,sec:explicit-finitary-projectors} in the domain
    theoretic language of
    finitary projectors;
  \item we build a model of \MLTTo{} in \(\dom\), interpreting types as
    finitary projectors, in \cref{sec:domain-model-model};
  \item we establish, in \cref{sec:logical-relation,sec:lr-properties},
    the adequacy of this model, using a logical relation;
  \item finally, in \cref{sec:lr-consequences} we derive injectivity of the \(\P\)
    constructor of \MLTTo{} from adequacy.
\end{itemize}

Although we use domain theory to explain the insight behind our proofs, it is
not strictly necessary: there is a purely elementary form of the proof, which
does not explicitly need any concepts from domain theory. Thus, this section
can be followed even without strong domain theoretic background.
We will in fact work at two levels at the same time: in the traditional language
of domain theory, as well as directly at the level of compact elements, which,
as explained above, capture the intuitive notion of finite observations. The
two approaches relate via an equivalence of categories
\begin{center}
  \begin{tikzcd}[column sep=small]
    \DomL \ar[rr, "\compacts{-}", bend left] & \eqv & \CUSL \ar[ll, "\ideals{-}", bend left]
  \end{tikzcd}
\end{center}
between the category \(\DomL\) of Scott domains and embedding-projection pairs as
morphisms, and
the category \(\CUSL\) of conditional upper semi-lattices and embeddings, see
\cref{def:scott-domain,def:cusl}.

Following a well-known recipe in domain theory, we will build an increasing sequence
of domains,
related by embedding-projection pairs
\[ \cdots \triangleleft \dom[n-1] \triangleleft \dom[n] \triangleleft \dom[n+1] \triangleleft
  \cdots
\]
approximating the solution \(\dom\) of the recursive equation. The sets of compact elements
\(\comp[n] = \compacts{\dom[n]}\) form an increasing sequence of \(\CUSL\), related
by inclusion
embeddings
\[ \cdots \subseteq \comp[n-1] \subseteq \comp[n] \subseteq \comp[n+1] \subseteq \cdots.
\]
In fact, the \(\comp[n]\) have an elementary inductive description,
which we use in our mechanisations.

\subsection{Computing a fixed point domain}
\label{sec:fix-solution}

We will build the model from a solution \(\dom\) to the following equation,
which represents \MLTTo{}'s constructors using functions to
capture binders. %
% , similarly to higher-order abstract syntax \cite{??},
The sum is the coalesced sum, which identifies the least elements,
and \([\dom \Rightarrow \dom]\) is the domain of continuous functions from
\(\dom\) to itself.
\begin{align*}
  \tag{dom} \label{eq:dom-eq}
  \dom &\quad \iso \quad \underbrace{\; [\dom \Rightarrow \dom] \; }_{\lambda} \; +
  \;\underbrace{\;\lift{(\dom \times [\dom \Rightarrow \dom])}\;}_{\Pi}\; +
  \underbrace{\quad\lift{1}\quad}_{\univ}
  % \\ &= [\dom \Rightarrow \dom] + \P \dom {[\dom \Rightarrow \dom]} + \univ
\end{align*}
Note that certain constructors are lifted, using \(\lift{}\), which adds
a least element to a domain, while the function space (encoding
λ-abstraction) is not. This means that \(\dom\) looks as follows:

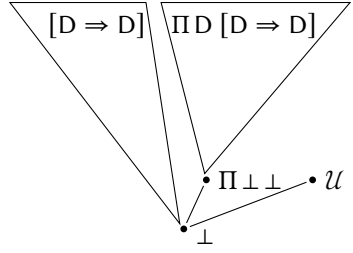
\begin{figure}[h]
  \begin{tikzpicture}[
      bullet/.style = {
        minimum size = 0pt,
        inner sep = .9pt,
        outer sep = 1.2pt,
        fill,
        circle
      },
      label/.style = {
        minimum size = 0pt,
        inner sep = 0pt
      },
      node distance = 15pt and 2pt,
    ]
    \node[bullet] (bot) {};
    \node[right=of bot, label, yshift=-3pt] {$\bot$};
    \draw (bot) -- (-2.3, 3) -- (-.5, 3) -- cycle;
    \node at (-1.15,2.7) {$[\dom \Rightarrow \dom]$};
    \node[bullet] (pi) [above right=15pt and 5pt of bot] {};
    \node[right=of pi, label] {$\Pi\,\bot\,\bot$};
    \node[bullet] (univ) [above right=15pt and 45pt of bot] {};
    \node[right=of univ, label] {$\univ$};
    \draw (bot) -- (pi);
    \draw (bot) -- (univ);
    \draw (pi) -- (-.3, 3) -- (2.2, 3) -- cycle;
    \node at (.8, 2.7) {$\P \dom\ {[\dom \Rightarrow \dom]}$};
  \end{tikzpicture}
  \caption{Representation of \(\dom\)}
  \label{fig:d}
  \Description{Coalesced sum merging the \(\bot\) element of \([\dom \Rightarrow \dom]\) with lifted
  parts of the rest of the sum.}
\end{figure}
The insight is that to validate η for functions we should
identify \(\bot\) and \(\l x. \bot\): both correspond to a function which diverges
on all arguments, thus no observation can distinguish them. On the
contrary, \(\Pi\,\bot\,\bot\) contains observable information: that the
type is a \(\Pi\)-type.

Standard results in domain theory \cite{Smyth1982, Scott1981} allow us to construct
a solution to
such an equation by interpreting it in terms of domains and continuous functions. To explain
how we work with \(\dom\), let us start with some domain theory primer.

\begin{definition}[Compatibility]
  Two elements \(u,v\) of a poset \(P\) are \emph{compatible},
  written \(u \compat_P v\), if there is \(w \in P\) such that \(u \le w\) and
  \(v \le w\).
\end{definition}

\begin{definition}[Directed and bounded subsets]
  A \emph{directed subset} of a poset \(P\) is a non-empty set \(X\) such that for any
  \(u,v \in X\), there exists \(w \in X\) such that \(u \le w\) and \(v \le w\).
  A subset \(X\) of a poset is \emph{bounded} if there is a \(v\) above all \(u \in X\).
\end{definition}

Viewing \(u \le v\) as meaning that “\(v\) contains more information than \(u\)”,
two elements are compatible when their common information can be subsumed
by a third one, and a directed set \(X\) is one where any two elements are
subsumed by a third element itself in \(X\).

\begin{definition}[Bounded and directed complete partial order]
  A partial order \(P\) is said to be \emph{directed complete} if every directed set
  \(X\) has a supremum, and \emph{bounded complete} if every bounded
  set \(X\) has a supremum. We write \(\bigvee X\) for the supremum of a set \(X\),
  and \(u \vee v\) for \(\bigvee \{u,v\}\).
\end{definition}

\begin{definition}[Compact elements]
  Let \(P\) be a directed complete poset.  An element \(u \in P\) is \emph{compact}
  if, for any directed \(X \subseteq P\), if \(u \le \bigvee X\), then there is some
  \(v \in X\) such that \(u \le v\).
\end{definition}

% A compact element \(u\) represents an undivisible amount of information, a primitive,
% basic observation: for any compound observation \(\bigvee X\) that is above \(u\),
% there is already an element \(v\) in \(X\) which contains the whole information
% that \(u\) contains.
A compact element represents a finite amount of information.
By contrast, a typical non-compact element is a function
\(\mathbb{N} \to \mathbb{N}\), because one cannot test every input/output pair
in finite time.

% \begin{definition}[Basis of an element]
%   Let \(P\) be a bounded and directed complete poset. For \(u \in P\), the \emph{basis} of $u$,
%   written $\basis{u}$, is the (bounded) set of compact elements below $u$.
% \end{definition}

\begin{definition}[Scott domains] \label[definition]{def:scott-domain}
  A \emph{Scott domain} $D$ is a bounded and directed complete poset $D$, such that
  every element \(u \in D\) is limit of its \emph{basis}, the
  (bounded) set of compact elements below it: \[
    u = \bigvee \{v \in \compacts{D} \mid v \le u \}
  \]
\end{definition}

In a Scott domain, we think of the order as the information order:
\(u \le v\) if \(u\) contains less information
(is less defined/diverges more/passes less observable tests).
The main property says that elements are completely characterised by the finite
pieces of information (\ie compact elements) they contain, or, viewed otherwise,
by the finite tests they pass.

Indeed, a Scott domain can in fact be reconstructed just from its set of compact elements.
\begin{definition}[Conditional upper semi-lattice \cite{StoltenberHansen1994}]
  \label[definition]{def:cusl}
  A \emph{conditional upper semi-lattice} (\(\CUSL\)) is a poset \(P\) with a least
  element \(\bot\) and
  where every two compatible elements \(u \compat_P v\) have a supremum \(u \vee v\).%
  \footnotemark
\end{definition}
\footnotetext{This is weaker than the common definition of a semi-lattice, which mandates that
\emph{every} two elements have a supremum.}

\begin{definition}[Ideals]
  Let \(P \in \CUSL\), an \emph{ideal} \(I\) of \(P\) is a subset of \(P\) that is downwards
  closed (if \(u \le v \in I\) then \(u \in I\)) and closed under finite suprema: if
  \(u,v \in I\), then \(u \vee v\) exists and \(u \vee v \in I\).  We write \(\ideals{P}\) for
  the poset of these ideals, ordered by inclusion.

  For any \(u \in P\) the downset of \(u\), \(\{v \in P \mid v \le u\}\),
  is an ideal, called the \emph{principal ideal} generated by \(u\).
  This embeds \(P\) in \(\ideals{P}\).
\end{definition}
% \footnotetext{The notation is here to suggest that \(\ideals{P}\) works as a
% completion of \(P\)
%   with ideal elements, where every element \(a \in P\) can be seen as a principal ideal
% \((a) \in \ideals{P}\).\mlb{TODO: change this.}}

\begin{proposition}
  The ideal completion of a poset \(P\) is a Scott domain if, and only if,
  $P$ is a $\CUSL$.
\end{proposition}
% \mc{TODO review propositions to clarify literature result vs formalized}

This means we can just as well work with a Scott domain or its \(\CUSL\) of
compact elements, since each can be reconstructed from the other.
This change of points of view can in fact be extended to an equivalence of categories.

\begin{definition}[Embedding]
  Let $C$ and $C'$ be two \(\CUSL\), a map \(f : C \to C'\) is an \emph{embedding} if:
  \begin{itemize}
    \item $f(\bot) = \bot$
    \item for $x, y \in C$, if $x \compat_C y$, then $f(x) \compat_{C'} f(y)$ and $f(x \vee
      y) = f(x) \vee f(y)$
    \item for $x, y \in C$, if $f(x) \leq f(y)$, then $x \leq y$;
    \item for $x, y \in C$, if $f(x) \compat_{C'} f(y)$, then $x \compat_C y$.
  \end{itemize}
\end{definition}
% \mlb{Either define both ep-pairs and embeddings (probably better) or neither?}

\begin{theorem}[\cite{StoltenberHansen1994}]
  The category $\DomL$ of Scott domains (and embedding-projection pairs as morphisms), and the
  category $\CUSL$ of conditional upper semi-lattices (and embeddings as morphisms) are
  equivalent. The map \(\DomL \to \CUSL\) maps a domain to its poset of compact elements, while
  the map \(\CUSL \to \DomL\) takes a conditional upper semi-lattice to the domain
  of its ideals.
\end{theorem}

Moreover, it is well-known that the category of Scott domains and continuous functions is
cartesian closed.  This means that every pair of Scott domains \(D\) and \(E\) induces a Scott
domain \([D \Rightarrow E]\) of continuous maps.
Given \(X, Y \in \CUSL\), we would thus like to describe
the compact elements of \([\ideals{X} \Rightarrow \ideals{Y}]\).

\begin{definition}[Step function]
  A \emph{step function} is a monotone function \(f : X \rightarrow Y\) such that
  \[
    f(u) = \bigvee \{ v_i \mid u_i \leq u \}
  \]
  for a finite set of pairs (which we suggestively write using \(\mapsto\))
  \[
    \{u_1 \mapsto v_1, \ldots, u_n \mapsto v_n\} \in \mathcal{P}_\mathsf{fin}(X \times Y).
  \]
  Such a finite set is called a \emph{presentation} of the step function \(f\).
\end{definition}

The presentation of a step function is not unique! For instance, we can always add
a pair \((u \mapsto \bot)\) to a presentation without changing the underlying
function. Conversely, not all sets of pairs represent a valid
function: a necessary and sufficient condition for a set
\(\{u_1 \mapsto v_1, \ldots, u_n \mapsto v_n\}\) to represent a function is that
\(u_i \compat_{X} u_j\) implies \(v_i \compat_{Y} v_j\), since
this ensures that the supremum \(\bigvee \{ v_i \mid u_i \leq u \}\) is always well
defined.

\begin{proposition}[Step functions are the compact functions]
  \label{prop:step-compact}
  The set \(X \Rightarrow Y\) of step functions, ordered pointwise, is a \(\CUSL\).
  Moreover, for every two Scott domains \(D\) and \(E\), we have
  \[
    \compacts{[D \Rightarrow E]} \quad \cong \quad \compacts{D} \Rightarrow \compacts{E}
  \]
\end{proposition}

Conversely, this means that for every pair of \(\CUSL\) \(X\) and \(Y\), we have
\[
  \ideals{X \Rightarrow Y} \quad \cong \quad [\ideals{X} \Rightarrow \ideals{Y}]
\]

\begin{proposition}
  The \(\CUSL\) of these step functions forms a functor
  \[ - \Rightarrow -\ :\ \CUSL \times \CUSL \longrightarrow \CUSL
  \]
\end{proposition}
Note that, unlike the usual function space,
this functor is covariant \emph{in both arguments}! It moreover preserves directed
colimits, which will play an important role in the construction below.

Our definition is arguably simpler than that of \textcite{StoltenberHansen1994}:
we define compact functions directly as step functions on compact elements,
rather than by a quotient.

This point of view of compact elements provides a way to solve equations such as
\cref{eq:dom-eq} \cite{Larsen91, StoltenberHansen1994}, by working directly with
\(\CUSL\).  One
first constructs the least fixed point in \(\CUSL\) of the functor (once again,
the sum is coalesced)
\begin{equation}
  F(X) \coloneq (X \Rightarrow X) + \lift{(X \times (X \Rightarrow X))} + \lift{1}
  \tag{func} \label{eq:func}
\end{equation}
as the colimit \(\comp \coloneq \bigcup_{n \in \mathbb{N}} \comp[n]\), with
\(\comp[0] \coloneq \{\bot\}\) and \(\comp[n+1] \coloneq F(\comp[n])\), and then defines
\(\dom \coloneq \ideals{\comp}\), which implies that \(\comp \cong \compacts{\dom}\).%
\footnote{In
  our case, \(\comp[n]\) is finite, and therefore \(\dom[n] \coloneq \ideals{\comp[n]}\) is
  isomorphic to \(\comp[n]\) (every ideal is principal).  However, when extending
  the type theory
  (for instance with an infinite hierarchy of universes), \(\comp[n]\) might not necessarily be
  finite, and then \(\dom[n]\) is no longer isomorphic to \(\comp[n]\).  Thus, we will not rely
on this property in the forthcoming proofs.}
Because \(F\) is covariant and continuous, \(\comp\) defined this way is a solution
to \cref{eq:func}. Combined with \cref{prop:step-compact}
(and the similar fact that the equivalence preserves
\(\lift{}\)), this implies that \(\dom\) is the solution to \cref{eq:dom-eq} we are after,
although we will mostly work directly with \(\comp\).

Unravelling what this construction concretely does, the elements of \(\comp[n+1]\)
are as follows, with \(\cuniv\), \(\cabs\) and \(\cpi\) corresponding
to the three components of the sum in \cref{eq:func}:
\begin{itemize}
  \item \(\bot\)
  \item \(\cuniv\)
  \item \(\cabs(f)\) for \(f\) a step function in \(\comp[n] \Rightarrow \comp[n]\)
  \item \(\cpi(a,b)\) for \(a \in \comp[n]\) and \(b \in \comp[n] \Rightarrow \comp[n]\)
  \item under the constraint that \(\cabs(\bot) = \bot\)
\end{itemize}
Thus \(\comp\) is essentially a quotient inductive type, with the quotient identifying
\(\cabs(\bot)\) with \(\bot\), and all different representations of the same
step function. And in fact one can forget most of the rest of this section and
simply work with this inductively defined \(\comp\).

\begin{proposition}
  The inclusion \(\comp[n] \subseteq \comp[n+1]\) is an embedding.
\end{proposition}

% Elements of \(\comp\) too are of one of the four forms above, with the compact functions in
% \(\comp^{\comp}\).

Finally, as we said \(\dom\) is obtained as the set \(\ideals{\comp}\) of ideals of \(\comp\),
and all operations can be lifted to \(\dom\), for instance \(\duniv \coloneq \{\bot,
\cuniv\}\) and,
if \(f \in [\dom \Rightarrow \dom] \iso \ideals{\comp \Rightarrow \comp}\), then
\(\dabs(f) \coloneq
\{\cabs(u)\mid u \in f\}\).

We can similarly define operations corresponding to eliminators. For instance,
we define the application operator \(\capp \in \comp \to \comp \to \comp\)%
\footnote{This is not a step function!}
by \(\capp(\cabs f,u) = f(u)\) and \(\capp(v,u) = \bot\) otherwise, then
derive \(\dapp \in \dom \to \dom \to \dom\) as
\[\dapp(\overline{v},\overline{u}) \coloneq \{\capp(v,u) \mid v \in \overline{v}, u \in
\overline{u}\}\]

\subsection{Semantic Typing}
\label{sec:semantic-typing}

The domain \(\dom\) will give us a good structure in which we can interpret terms of
\MLTTo{}. However, before moving to that model we also need to explain how to interpret the
typing relation. To do this, we define a semantic typing relation on elements of
\(\comp[n]\), in
\cref{fig:sem-typing}.

\begin{figure}
  \begin{mathpar}
    \jform{\strut u \ty_{n} a}[\(u\) has type \(a\) (\(u, a \in \comp[n]\))]
    \inferdef[Bot]{a \ty_{n} \cuniv}{\bot \ty_{n} a} \and
    \inferdef[Univ]{ }{\cuniv \ty_{n+1} \cuniv} \and
    \inferdef[Lam]{
      a \ty_n \cuniv \\
      \domain{f}{a} \\
      \forall u \in \comp[n], f(u) \ty_n g(u)
    }{f \ty_{n+1} \cpi(a,g)} \and
    \inferdef[Pi]{
      a \ty_n \cuniv \\
      \domain{g}{a} \\
      \forall u \in \comp[n], g(u) \ty_n \cuniv
    }{\cpi(a, g) \ty_{n+1} \cuniv} \\
    \domain{f}{a} \coloneq
    \forall u \in \comp[n], \exists v \in \comp[n], v \leq u \land v \ty_{n} a
    \land f(u) = f(v)
  \end{mathpar}
  \caption{Semantic typing rules}
  \label{fig:sem-typing}
  \Description{Semantic typing rules, case splitting on the type.}
\end{figure}

To handle functions, we use an auxiliary predicate which asserts that the domain of a (compact)
function is a given type, in the sense that the function is completely characterised by its
values on elements of that type.  This extensional (negative) presentation of typing is
equivalent to an intensional (positive) one, thanks to the following characterisation.

\begin{proposition}[Characterisation of function typing]
  \label[proposition]{thm:fun-typing}
  Assume given \(f \in \comp[n] \Rightarrow \comp[n]\) and \(a \in \comp[n]\).  The predicate
  \(\domain{f}{a}\) holds if and only if there exists
  \(\{u_i \mapsto v_i\} \in \mathcal{P}_\mathsf{fin}(\comp[n]\times\comp[n])\)
  representing \(f\)
  such that \(u_i \ty_{n} a\) for all \(i\).

  Moreover, if \(g \ty_n \cpi(a,\_\mapsto\cuniv)\), we have that \(\forall u \in
  \comp[n], f(u) \ty_n g(u)\) if, and only if, there exists a representation \(\{u_i
  \mapsto v_i\}\) of \(f\) such that
  \[
    \forall i, v_i \ty_n g(u_i)
  \]
\end{proposition}

Importantly, the properties only hold for \emph{some well-chosen} representation of \(f\)!
Indeed, one can extend any representation with a pair \((u,\bot)\)
without changing the represented function, but this \(u\) might not be of type \(a\).
Although we will not make use of this fact, there is actually a canonical representation
of functions, which will always satisfy the intensional properties whenever
the function satisfies the extensional ones.

Semantic typing has a number of good closure properties.
% : terms of a given type
% are closed downwards and under (conditional) supremum,
% and typing is preserved and reflected by the embedding
% inclusion of \(\comp[n]\) in \(\comp[n+1]\).

\begin{proposition}
  \label[proposition]{thm:typing-lub}
  Suppose that \(u, v \in \comp[n]\) are compatible.  In that case, for all \(a \in \comp[n]\)
  we have \[
    u \ty_n a \quad \text{and} \quad v \ty_n a \quad \Longrightarrow \quad (u \vee v) \ty_n a
  \]
\end{proposition}

\begin{proposition}
  \label[proposition]{thm:typing-up}
  Let \(u, a, a' \in \comp[n]\).  We have that \[
    u \ty_n a \quad \text{and} \quad a \leq a' \quad \Longrightarrow \quad u \ty_n a'
  \]
\end{proposition}

\begin{proposition}[Stability of semantic typing]
  \label[proposition]{thm:typing-inj}
  For every \(u, a \in \comp[n]\), it holds \[
    u \ty_n a \quad \Longleftrightarrow \quad u \ty_{n+1} a
  \]
\end{proposition}

This hierarchy of semantic typing on \(\comp[n]\) induces a notion of semantic typing on
\(\comp\), defined as follows.  Given \(u,a \in \comp\), we define
\[ u \ty a \quad \coloneq \quad \exists n, u \ty_n a.
\]
Note that, by definition of \(\comp = \bigcup_{n \in \mathbb{N}} \comp[n]\), there exists, for
all \(u, a \in \comp\), a least \(n\), which can be computed, such that \(u, a \in
\comp[n]\).  It follows, from \cref{thm:typing-inj}, that \[
  u \ty a \quad \Longleftrightarrow \quad u \ty_{n} a.
\]
In the case where \(u \ty a\), we have moreover that \(u \ty_m a\) for any \(m \geq n\).

We can finally extend semantic typing to all domain elements \(\overline{u},
\overline{a} \in \dom\),
in the following way:
\[
  \overline{u} \ty \overline{a} \quad \coloneq \quad \forall u \in \overline{u},
  \exists v \in \overline{u}, \exists a \in \overline{a}, u \leq v \land v \ty a.
\]
Both typing relations, on \(\comp\) and on \(\dom\), inherit the good properties of
\(\_ \ty_{n} \_\).

\subsection{Finitary projectors}
\label{sec:finitary-projectors}

We can in fact analyse semantic typing using the well established concept
of finitary projector on a Scott domain. For the rest of this
section we fix a Scott domain \(D\) and set \(C \coloneq \compacts{D}\).
\begin{definition}[Finitary projector \cite{Scott1981}]
  \label[definition]{def:finitary-projector}
  A \emph{finitary projector} on \(D\) is a continuous map
  $p \in [D \Rightarrow D]$ such that
  \begin{enumerate}
    \item for $\alpha \in D$, $p\,\alpha \leq \alpha$
    \item for $\alpha \in D$, $p(p\,\alpha) = p\,\alpha$
    \item the image of $p$ is a Scott domain with respect to the restriction of the
      order on~$D$.
  \end{enumerate}
\end{definition}

We define \(\Proj D\) as the set of finitary projectors of \(D\),
ordered pointwise. An important observation of Scott is that this is again a domain.

\begin{theorem}[\cite{Scott1981}]
  \(\Proj D\) is a Scott domain.
\end{theorem}

We use the notation $D_{|p}$ for the image of~$p$, that is \[
  D_{|p} \coloneq \{p(\alpha) \mid \alpha \in D\}.
\]
An element $\alpha\in D$ is in the image of $p$ precisely when $p(\alpha)=\alpha$.
Thus, the Scott domain \(D_{|p}\) is the set of
fixed points of \(p\). Accordingly, we use the notation
\(C_{|p}\) for the set
of compact fixed points of \(p\):
\[ C_{|p} \coloneq C \cap D_{|p} = \{ u \in C \mid p(u) = u \}
\]

\begin{proposition}
  Given any finitary projector \(p\) on \(D\),
  an element $u \in D_{|p}$ is compact in~$D$ if and only if it is compact in~$D_{|p}$.
  This means that \[
  C_{|p} = \compacts{D_{|p}} \]
\end{proposition}

It follows that:

\begin{proposition}
  Every fixed point $\alpha\in D_{|p}$ is a limit of compact fixed points in~$C_{|p}$.
\end{proposition}
\begin{proof}
  We use the fact that, by definition, the image \(D_{|p}\) is a Scott domain.  From
  this follows
  that every element in \(D_{|p}\) is the limit of the compact elements below it.
  One concludes
  with the characterisation of the compact elements of \(D_{|p}\) as the elements of
  \(C_{|p}\).
\end{proof}
This property provides an alternative formulation of the fact the image \(D_{|p}\)
of a finitary
projector is a Scott domain in \cref{def:finitary-projector}.  This, together with
continuity,
implies that every finitary projector~$p$ can be recovered by its
set of compact
fixed points $C_{|p}$ in the following way.
\begin{proposition}
  The set of compact fixed points of a finitary projector~$p \in [D \Rightarrow D]$
  is closed under finite suprema.
\end{proposition}
Conversely, a simple argument establishes that, for \(C = \compacts{D}\):
\begin{proposition}\label[proposition]{thm:proj-comp-set}
  Every set~$E \subseteq C$ closed under finite sups defines a finitary projector
  \[ p(\alpha) = \{ u \in E \mid u \leq \alpha \}.
  \]
  This finitary projector is uniquely determined by the equation $E=C_{|p}$.
\end{proposition}

\subsection{Semantic types expressed as finitary projectors}
\label{sec:explicit-finitary-projectors}

In this section we explain how to use curryfication to view each level \(- \ty_n -\) of the
semantic typing hierarchy as a continuous function \(\p_n : \dom[n] \rightarrow
\Proj \dom[n]\).
The function \(\p_n\) is defined as
\[
  \arraycolsep=2pt
  \begin{array}{r c c l l}
    \p_n & : & \dom[n] & \longrightarrow & \Proj \dom[n] \\
    & & \alpha &\longmapsto & \{ u \in \comp[n] \mid \exists a \in \alpha, u \ty_n a \}
  \end{array}
\]
where we identify, thanks to \cref{thm:proj-comp-set}, the set of compact elements with the
finitary projector defined as
\[ \p_n(\alpha) = \omega \mapsto \bigvee \{ u \in \comp[n] \mid u \in \omega \land
    \exists a \in
  \alpha, u \ty_n a \}
\]
\begin{proposition}
  The map $\p_n$ is well-defined, and continuous.
\end{proposition}
\begin{proof}
  The set \(\{ u \in \comp[n] \mid \exists a \in \alpha, u \ty_n a \}\) is closed
  by conditional sup.  Indeed, let \(u_1, u_2 \in \comp[n]\) compatible such that, for
  \(i \in \{1, 2\}\), there exists \(a_i \in \alpha\) such that \(u_i \ty_n a_i\).  Because
  \(\alpha\) is an ideal, \(a_1 \compat_{\comp[n]} a_2\) and \(a_1 \vee a_2 \in
  \comp[n]\).  By
  \cref{thm:typing-up}, we have
  \[ u_1 \ty_n (a_1 \vee a_2) \qquad \text{and} \qquad u_2 \ty_n (a_1 \vee a_2)
  \]
  and by \cref{thm:typing-lub}, we conclude that \[
    (u_1 \vee u_2) \ty_n (a_1 \vee a_2)
  \]
  thus \[
    u_1 \vee u_2 \in \{ u \in \comp[n] \mid \exists a \in \alpha, u \ty_n a \}.
  \]
  Therefore, by \cref{thm:proj-comp-set}, this set is the set of compact fixed
  points of a finitary projector.

  Monotonicity is immediate, and continuity is a consequence of a fact that the supremum of a
  directed set \(\{p_i\}\) of finitary projectors in \(\Proj \dom[n]\), is the
  finitary projector
  \(p = \bigvee_{i} p_i\) whose set of compact fixed points \(C_{|p}\) is precisely
  the union of
  the \(C_{|p_i}\).
\end{proof}

There is a one-to-one relationship between:
\begin{itemize}
  \item finitary projectors on \(D\);
  \item embedding-projection pairs \(D' \triangleleft D\), up to isomorphism.
\end{itemize}

Every embedding-projection \((e, r) : D \triangleleft E\) induces a function
\(e_* : \Proj D \to \Proj E\), which transports every finitary projector \(p \in \Proj D\) to \[
  e_*(p) \coloneq e \circ p \circ r.
\]

Conversely, an embedding-projection \((e, r) : D \triangleleft E\) also induces a function
\(e^* : \Proj E \to \Proj D\) which transports every finitary projector \(q \in \Proj E\) to the
finitary projector whose set of compact fixed points is
\(\compacts{D} \cap e^{-1}(\compacts{E}_{|q})\),
which happens to be a Scott domain.

Note that \(e^*\) and \(e_*\) define an embedding-projection: if \(D \triangleleft E\), then
\(\Proj D \triangleleft \Proj E\).  This shows that \(\Proj : \DomL \to \DomL\) defines a
functor.

The property that \(u \ty_n a \iff u \ty_{n+1} a\) can be expressed in the following
diagrammatic way:
\begin{proposition}
  \begin{center}
    \begin{tikzcd}
      \dom[n+1] \ar[r, "\p_{n+1}"] & \Proj \dom[n+1] \ar[d, "e_n^*"] \\
      \dom[n] \ar[u, "e_n"] \ar[r, "\p_n"] & \Proj \dom[n]
    \end{tikzcd}
  \end{center}
\end{proposition}

Recall the application operator \(\dapp \in \dom \to \dom \to \dom\) from
\cref{sec:fix-solution}; concretely, \(\dapp(\dabs f,u) = f(u)\) and
\(\dapp(v,u) = \bot\) otherwise.  By continuity of \(\Proj\), the maps \(\p_n\) lift
to a single \(\p \colon \dom \rightarrow \Proj \dom\), satisfying the following
equations:
\begin{align*}
  \p\;\duniv\;\duniv &= \duniv \\
  \p\;\duniv\;(\dpi(\alpha, \beta))
  &= \dpi (\p\;\duniv\;\alpha, u \mapsto
  \p\;\duniv(\beta(\p\;\alpha\;u))) \\
  \p\;(\dpi(\alpha, \beta))\;f &= u\mapsto
  \p\;(\beta\;(\p\;\alpha\;u))\;(\dapp(f,\p\;\alpha\;u)) \\
  \p\;\alpha\;\beta &= \bot & \text{in all other cases.}
\end{align*}

Composing the embeddings \(e_n\) (resp.\ the projections \(r_n\)) yields, for each
\(n\), the injection of \(\dom[n]\) into the colimit and the corresponding projection,
\[
  e_n^\infty \colon \dom[n] \to \dom, \qquad r_n^\infty \colon \dom \to \dom[n],
\]
which again form an embedding--projection pair:
\(r_n^\infty \circ e_n^\infty = \mathrm{id}_{\dom[n]}\) and
\(\bigvee_n e_n^\infty \circ r_n^\infty = \mathrm{id}_{\dom}\).

Using these maps, the commutation of the diagram above can be refined into the
following pointwise identity, relating the global projector \(\p\) to its finite
stages.
\begin{proposition}
  For all \(a, x \in \dom[n]\),
  \[
    r_n^\infty\big(\p\;(e_n^\infty\,a)\;(e_n^\infty\,x)\big) = \p_n\;a\;x.
  \]
\end{proposition}

This gives an alternative to the definitions of \(\p_n\) and of the typing relation: one may
first define \(\p \colon \dom \to \dom \to \dom\) directly by the recursive
equations above, then
\emph{define} \(\p_n\) from \(\p\) by
\(\p_n\;a\;x \coloneq r_n^\infty\big(\p\;(e_n^\infty\,a)\;(e_n^\infty\,x)\big)\), and finally
prove, by induction on \(n\), that each \(\p_n\,a\) is a finitary projection.

\subsection{The Domain Model}
\label{sec:domain-model-model}

\begin{figure}
  \begin{mathpar}
    % \jform{\typing{\Gamma}{A}}[\(A\) is a type in \(\Gamma\)]
    % %
    % \inferdef[FunTy]{\typing{\Gamma}{A} \\ \typing{\Gamma, x : A}{B}}
    % {\typing{\Gamma}{\P x : A.B}} \and
    % \inferdef[UnivTy]{\ctxty{\Gamma}}{\typing{\Gamma}{\univ}} \\
    \jform{\typing{\Gamma}{M}[A]}[\(M\) has type \(A\) in \(\Gamma\)]
    \inferdef[Univ]{\ctxty{\Gamma}}{\typing{\Gamma}{\univ}[\univ]} \and
    \inferdef[Pi]{\typing{\Gamma}{A}[\univ] \\ \typing{\Gamma, x :
    A}{B}[\univ]}
    {\typing{\Gamma}{\P_{x : A}B}[\univ]} \and
    \inferdef[Lam]{\typing{\Gamma}{A}[\univ] \\ \typing{\Gamma, x : A}{M}[B]}
    {\typing{\Gamma}{\l_{x : A}M}[\P_{x : A}B]} \and
    \inferdef[App]{\typing{\Gamma}{M}[\P_{x : A}B] \\ \typing{\Gamma}{N}[A]}
    {\typing{\Gamma}{M \, N}[\subs{B}{N}]} \\
    \jform{\conv{\Gamma}{M}{N}[A]}[\(M\) and \(N\) are definitionally equal at
    type \(A\) in \(\Gamma\)]
    \inferdef[β]{
      \typing{\Gamma}{A}[\univ] \\
      \typing{\Gamma, x : A}{M}[B] \\
      \typing{\Gamma}{N}[A]
    }{
      \conv{\Gamma}{(\l_{x : A}M)\,N}{\subs{M}{N}}[\subs{B}{N}]
    } \label{rule:beta-fun} \and
    \inferdef[η]{\conv{\Gamma, x : A}{M\,x}{N\,x}[B]}
    {\conv{\Gamma}{M}{N}[\P_{x : A}B]}
    \label{rule:eta-fun}
  \end{mathpar}

  \caption{Selected typing rules of \MLTTo{}}
  \label{fig:syntax}
  \Description{Typing and conversion rules for a version of MLTT with a Π type and
    a universe;
  the single universe types itself, and we have β and η rules for conversion of functions.}
\end{figure}

The next step is to use the structure built in \cref{sec:fix-solution} to
construct a model of
dependent type theory, and relate it to the semantic typing of
\cref{sec:semantic-typing}. The
syntax of our type theory, \MLTTo{}, is in \cref{fig:syntax} --~we omit the definition of
substitution \(\subs{M}{N}\), which substitutes the last variable in scope, and
congruence rules for definitional equality. We write \(\tm\) for the set of untyped terms,
and interpret these directly as elements of \(\dom\) –~ill-typed terms will simply
be interpreted using \(\bot\).

\begin{definition}[Environments]
  A \emph{domain environment} \(\rho \in \env_{\dom}\)
  is a mapping of variables to elements of \(\dom\), and a \emph{compact environment}
  \(\rho \in \env_{\comp}\) is a mapping of variables to elements of \(\comp\).
  The \emph{extension} \((\rho, x \mapsto a)\) maps \(x\) to \(a\) and any other
  variable \(y\) to
  \(\rho(y)\). Environments are ordered pointwise, and have a least element
  \(\bot \coloneq \_ \mapsto \bot\).
\end{definition}

With the idea that domain elements are ideals of compact elements, we can directly
describe the semantics of a term (by induction on it) as a set of compact elements,
which we then
show to be an ideal, see the left side of \cref{fig:sem-compacts}.
Alternatively, we
can describe the semantics directly in terms of domain elements, making it clear that the
definition is the natural one. This characterisation is on the right side of
\cref{fig:sem-compacts}. More precisely if \(\rho \in \env[\dom]\), we have
\[\bigcup_{\rho' \in \env[\comp], \rho' \le \rho} \sem{t}[\rho'] \quad = \quad \sem{t}[\rho]\]
where the left semantics is in terms of \(\ideals{\comp}\) and the right is in
terms of \(\dom\).

\begin{figure}
  \setlength{\arraycolsep}{6pt}
  \renewcommand{\arraystretch}{1.5}
  \[
    \begin{array}{rl|rl}
      \sem{\_}[\_]           \in & \tm \to \env[\comp] \to \ideals{\comp} &
      \sem{\_}[\_]           \in & \tm \to \env[\dom] \to \dom \phantom{\ideals{C}}
      \\[.5em]

      \sem{x}[\rho] \coloneq & \{u \in \comp \mid \rho(x) \le u\} &
      \sem{x}[\rho] = & \rho(x) \\

      \sem{\univ}[\rho]      \coloneq & \{\bot, \cuniv\} &
      \sem{\univ}[\rho]      = & \duniv \\

      \sem{\P_{x : A}B}[\rho] \coloneq &
      \left\{\cpi(a,b) \mid a \in \sem{A}[\rho] \wedge \domain{b}{a} \right. &
        \sem{\P_{x : A}B}[\rho] = &
        \dpi(\sem{A}[\rho], \\[-.2em]
          & \hfill \left.
          \wedge~\forall_{u : a} b~u \in \sem{B}[\rho, x
        \mapsto u]\right\}
        \cup \{\bot\}
      && \hfill u \mapsto \sem{B}[\rho, x \mapsto \p{\sem{A}[\rho]}{u}])  \\

      \sem{\l_{x : A}M}[\rho] \coloneq &
      \left\{\cabs(f) \mid
        \exists_{a \in \sem{A}[\rho]}\left(\domain{f}{a} \right.\right. &
          \sem{\l_{x : A}M}[\rho] = & \dabs(u \mapsto \sem{M}[\rho, x
          \mapsto \p{\sem{A}[\rho]}{u}]) \\[-.4em]
          & \hfill \left.\left.
          \wedge~\forall_{u : a} f\,u \in \sem{M}[\rho, x \mapsto
      u]\right)\right\} \\

      \sem{M\,N}[\rho]       \coloneq &
      \left\{ \capp(u,v) \mid u \in \sem{M}[\rho] \wedge v \in \sem{N}[\rho]\right\} &
      \sem{M\,N}[\rho] = & \dapp(\sem{M}[\rho],\sem{N}[\rho])
    \end{array}
  \]

  \caption{Definition of the semantics of a term as a predicate on compacts, and
  description as a domain element}
  \label{fig:sem-compacts}
  \Description{On the right, the definition in terms of domain elements is
  straightforward: every syntactic operation is mapped to the corresponding
  domain operation. On the left, the definition in terms of compact elements
  uses subsets and the \(\dom\) predicate to describe the ideals.}
\end{figure}

\begin{lemma}[Semantic substitution]
  For any terms \(M\) and \(N\), and compact environment \(\rho\), \[
    \sem{M[N/x]}[\rho] = \sem{M}[\rho,x\mapsto\sem{N}[\rho]]
  \]
\end{lemma}
\begin{proof}
  By structural induction on \(M\).
\end{proof}

This semantics, a priori defined on untyped terms, is in fact sound, in that the
denotation of a well-typed term is well-typed.

\begin{definition}
  An environment \(\rho\) \emph{fits} a context \(\Gamma\) if for all \((x : A)
  \in \Gamma\) we
  have that \(\sem{A}[\rho] \ty \cuniv\) and \(\rho(x) \ty \sem{A}[\rho]\).
\end{definition}

\begin{theorem}[Soundness of the semantics
  \lean{Sound}{2900}{LE_Interp.sound}\,\rocq{finelt/typing_semantics}{1480}{typing_EvalRel}\,\agda{ID/Model/Soundness}{108}{theorem1}]
  \label[theorem]{thm:soundness}
  Assume \(\rho\) fits \(\Gamma\). Then
  \begin{itemize}
    \item \(\typing{\Gamma}{M}[A]\) implies that \(\sem{M}[\rho] \ty \sem{A}[\rho]\);
    \item \(\conv{\Gamma}{M}{N}[A]\) implies \(\sem{M}[\rho] = \sem{N}[\rho]\).
  \end{itemize}
\end{theorem}

Note that, thanks to the use of a coalesced sum to interpret λ-abstraction in
\cref{eq:dom-eq},
the model validates \(\dabs(\bot) = \bot\), which is critical to interpret
\ruleref{rule:eta-fun}.

\subsection{The Logical Relation}
\label{sec:logical-relation}

The model of \cref{sec:domain-model-model} already gives us some of the inversions
we are after.

\begin{theorem}[Disjointness of type constructors
  \lean{Consequences}{520}{sort_forallE_inv'}\,\rocq{finelt/adequacy}{2763}{tuniv_not_tpi}]
  For any (well-formed) context \(\Gamma\) and types \(\typing{\Gamma}{A}[\univ]\) and
  \(\typing{\Gamma,x:A}{B}[\univ]\),
  we have \(\nconv{\Gamma}{\P_{x : A}B}{\univ}[\univ]\).
\end{theorem}

\begin{proof}
  Assume \(\conv{\Gamma}{\P_{x : A}B}{\univ}[\univ]\). The environment
  \(\bot \coloneq \_ \mapsto \bot\) fits all contexts, thus by \cref{thm:soundness}
  \(\sem{\P_{x :A}B}[\bot] = \sem{\univ}[\bot]\) and, computing both semantics,
  \(\dpi\bigl(\sem{A}[\bot],\ u \mapsto \sem{B}[\bot, x \mapsto \p{\sem{A}[\bot]}{u}]\bigr) =
  \duniv.\) But in the model the \(\dpi\) and \(\duniv\) constructors are
  disjoint, so this is
  impossible.
\end{proof}

However, things are not so simple for injectivity properties, for two reasons.
First, instead of falsity we need to conclude definitional equality.
But even if the \(\dpi\) constructor is injective in the model, this will only
imply semantic equalities, from which it is not possible to go back to
the syntax. Second, if we wish to deduce injectivity in non-empty contexts,
we also need to take care of open terms, and thus neutrals. This is why we
will build a logical relation on top of our model.

\paragraph{Structure} The main specificity of our logical relation is the
structure on which it is built. In the absence of dependent types,
a logical relation is a predicate/relation on terms defined by recursion on
types. With dependent types, a naïve version of this fails, because types
themselves compute and thus one cannot simply define the predicate by recursion
on it. Instead, the standard solution is to do recursion on a semantic object
associated to each type –~its normal form. Here, we do not have normal forms,
but domain elements play a similar role.

Concretely, we define two relations
\(\TValEq{\Gamma}{A}{A'}{a}\) and \(\ValEq{\Gamma}{M}{M'}{A}{u}{a}\), where
\(A\), \(A'\), \(M\) and \(M'\) are syntactic terms in a common context \(\Gamma\),
and \(a, u \in \comp[n]\).
%% Both relations are moreover parametrised by an
%% environment \(\rho\) that \emph{fits} \(\Gamma\), that is, interprets each free
%% variable of \(\Gamma\) by a finite element; we keep \(\rho\) implicit in the
%% notation.\mc{This is technically not true; in the formalization there is no such
%%   $\rho$. Conceptually we could say that one exists but we either always use $\Gamma=[]$
%%   or otherwise for nonempty $\Gamma_0$ we just use $\rho=\bot$. This base context never
%%   changes even when we enter a binder, so we don't need to provide interpretations for
%% the free variables in it.}
The intent is that \(A\) and \(A'\)
(resp.\ \(M\) and \(M'\)) are related types (resp.\ terms of \(A\)),
where the relation is constructed by following the structure of \(a\) and \(u\).
As precondition for this relation to make sense, we should have \(u \ty_{n} a\),
\(a\) should be in the interpretation of \(A\) and \(A'\),
and \(u\) in that of \(M\) and \(M'\). More precisely, we should have
\(a \le \sem{A}[\bot]\) and \(a \ty_n \cuniv\) (and similarly for the other terms).

Variables in $\Gamma$ are all interpreted by $\bot$, and so neutral
terms%
\footnote{Eliminations of variables, such as \(x\,0\).}
will be interpreted by $\bot$ as well. This is a sort of compromise between
closed and open logical relations: we still work with terms in an open context,
and will thus be able to deduce definitional inversion properties in an open
context too; but because we interpret neutrals in this degenerate fashion,
we can avoid the weight of Kripke-style quantifications.
This however means that our technique does not directly apply to obtain
definitional inversion properties for neutrals, such as the fact that
\(x~M \convop y~N\) implies that \(x\) and \(y\) are the same variable, and
that \(M \convop N\). We come back to this issue in the conclusion.
\begin{comment}
which are defined under the following constraints
\begin{align*}
  \TVal{\Gamma}{A}{a} &\Rightarrow a \leq \sem{A}_n \\
  \Val{\Gamma}{M}{A}{u}{a} &\Rightarrow
  \begin{cases}
    \typing{\Gamma}{M}[A] \\
    u \leq \sem{M}_n \\
    a \leq \sem{A}_n \\
    u \ty_n a
  \end{cases} \\
  \TValEq{A_1}{A_2}{a} &\Rightarrow a \leq \sem{A}_n \\
  \ValEq{M_1}{M_2}{A}{u}{a} &\Rightarrow
  \begin{cases}
    a \leq \sem{A}_n \\
    u \leq \sem{M}_n \\
    u \ty_n a
  \end{cases}
\end{align*}
\end{comment}

\paragraph{Head reduction} The relations are stated up to head reduction. We
write \(M \sred M'\) for a single \emph{head-reduction} step, generated by
\(\beta\)-reduction and its congruence in function position,
\begin{mathpar}
  (\l_{x : A}M)\,N \sred \subs{M}{N} \and
  \inferdef{M \sred M'}{M\,N \sred M'\,N}
\end{mathpar}
and \(M \red M'\) for its reflexive--transitive closure.  The \emph{typed} head
reduction \(\convred{\Gamma}{M}{M'}[A]\) is defined to hold when \(M \red M'\)
\emph{and} \(\conv{\Gamma}{M}{M'}[A]\).  We build the conversion into the
definition because subject reduction is not yet available at this point, so it
cannot be recovered from \(M \red M'\) and \(\typing{\Gamma}{M}[A]\) alone.

\paragraph{Definition} We abbreviate the diagonals of the two relations by
\[
  \TVal{\Gamma}{A}{a} \defeq \TValEq{\Gamma}{A}{A}{a}
  \qquad\text{and}\qquad
  \Val{\Gamma}{M}{A}{u}{a} \defeq \ValEq{\Gamma}{M}{M}{A}{u}{a}.
\]
The four relations \(\TVal{\Gamma}{A}{a}\), \(\TValEq{\Gamma}{A}{A'}{a}\),
\(\Val{\Gamma}{M}{A}{u}{a}\) and \(\ValEq{\Gamma}{M}{M'}{A}{u}{a}\) can be seen as
refinements of the four judgements of dependent type theory,
\(\typing{\Gamma}{A}[\univ]\), \(\conv{\Gamma}{A}{A'}[\univ]\),
\(\typing{\Gamma}{M}[A]\) and \(\conv{\Gamma}{M}{M'}[A]\) respectively, each
annotated with a semantic witness \(a, u \in \comp[n]\) recording the shape of
the (in)equality.

As explained, the definition goes by induction on the stage \(n\), and then by
case analysis on the semantic type \(a\).
In the base case where it is \(\bot\), the relation is
degenerate.
\begin{equation*}
  % \TVal{\Gamma}{A}{\bot} &\coloneq \T \\
  \TValEq{\Gamma}{A}{A'}{\bot} \defeq \T \hspace{.2\textwidth}
  % \Val{\Gamma}{M}{A}{u}{\bot} &\coloneq \T \\
  \ValEq{\Gamma}{M}{M'}{A}{u}{\bot} \defeq \T
\end{equation*}
This is very reminiscent of step-indexed logical relations,
which have a similar degenerate base case. Indeed we can think of this as a
generalization of step-indexing where the step indexes are replaced by a partial
order. We will not explore the connection further in this work.

In the case where \(a\) is a Π-type, two types are related if they both head-reduce
to syntactic Π-types with syntactically equal and semantically
related domains and codomains, the latter incorporating the presence of a
binder by requiring that the codomains map related inputs to related outputs.
\begin{align*}
  &\TValEq[n+1]{\Gamma}{T}{T'}{\cpi(a,b)} \quad\defeq\\
  &\hspace{1em}
  \begin{cases}
    \convred{\Gamma}{T}{\P_{x : A}B}[\univ] \qquad
    \convred{\Gamma}{T'}{\P_{x : A'}B'}[\univ] \\
    \conv{\Gamma}{A}{A'}[\univ] \qquad
    \conv{\Gamma,x:A}{B}{B'}[\univ] \qquad
    \TValEq{\Gamma}{A}{A'}{a} \\
    \forall M~M'~v,\ %
    \begin{cases}
      \presup{v \ty_n a} \quad
      \presup{\conv{\Gamma}{M}{M'}[A]} \\
      \ValEq{\Gamma}{M}{M'}{A}{v}{a}
    \end{cases}
    \Rightarrow\hspace{.5em}
    \begin{cases}
      \TValEq{\Gamma}{\subs{B}{M}}{\subs{B}{M'}}{b(v)} \\
      \TValEq{\Gamma}{\subs{B'}{M}}{\subs{B'}{M'}}{b(v)} \\
      \TValEq{\Gamma}{\subs{B}{M}}{\subs{B'}{M'}}{b(v)}
    \end{cases}
  \end{cases}
\end{align*}
Terms are related at function types if they map related inputs to related outputs:
\begin{align*}
  &\ValEq[n+1]{\Gamma}{M}{M'}{T}{u}{\cpi(a,b)}\quad \defeq\\
  &\hspace{1em}
  \begin{cases}
    \presup{\TVal[n+1]{\Gamma}{T}{\cpi(a,b)}} \qquad
    \convred{\Gamma}{T}{\P_{x : A}B}[\univ] \\
    \forall N~N'~w,\ %
    \begin{cases}
      \presup{w \ty_n a} \quad
      \presup{\conv{\Gamma}{N}{N'}[A]} \\
      \ValEq{\Gamma}{N}{N'}{A}{w}{a}
    \end{cases}
    \Rightarrow\hspace{.5em}
    \begin{cases}
      \ValEq{\Gamma}{M\,N}{M\,N'}{\subs{B}{N}}{\capp(u,w)}{b(w)} \\
      \ValEq{\Gamma}{M'\,N}{M'\,N'}{\subs{B}{N}}{\capp(u,w)}{b(w)} \\
      \ValEq{\Gamma}{M\,N}{M'\,N'}{\subs{B}{N}}{\capp(u,w)}{b(w)}
    \end{cases}
  \end{cases}
\end{align*}

Finally, when \(a\) is the universe, two types are related if they both reduce
to \(\univ\), and two terms are related, at a type which itself reduces to
\(\univ\), if they are related as types.
\begin{align*}
  \TValEq{\Gamma}{A}{A'}{\cuniv} &\defeq
  \convred{\Gamma}{A}{\univ}[\univ] \quad \wedge \quad
  \convred{\Gamma}{A'}{\univ}[\univ] \\
  \ValEq[n+1]{\Gamma}{M}{M'}{A}{u}{\cuniv} &\defeq
  \convred{\Gamma}{A}{\univ}[\univ] \quad \wedge \quad
  \TValEq[n+1]{\Gamma}{M}{M'}{u}
\end{align*}

\subsection{Properties of the logical relation}
\label{sec:lr-properties}

We state the properties for the two-sided relations \(\TValEq{\Gamma}{A}{A'}{a}\) and
\(\ValEq{\Gamma}{M}{N}{A}{u}{a}\); the statements for the diagonals
\(\TVal{\Gamma}{A}{a}\) and
\(\Val{\Gamma}{M}{A}{u}{a}\) are the special cases where the two syntactic
arguments coincide.

Recall from the discussion above that the relations carry their own well-formedness
conditions: \(\ValEq{\Gamma}{M}{N}{A}{u}{a}\) presupposes \(u \ty_n a\) and
\(a \ty_n \cuniv\), and \(\TValEq{\Gamma}{A}{A'}{a}\) presupposes \(a \ty_n \cuniv\).
Consequently, some of the premises below are redundant: they already follow from one
of the \(\TValEq{}{}{}{}\) or \(\ValEq{}{}{}{}{}{}\) hypotheses appearing in the same
rule.  We typeset such a premise in \presup{grey}.  Grey premises may be read as
mere \emph{presuppositions}, carried along for the record; only the black ones
constitute genuine proof obligations.  We nonetheless spell them out, because the
mechanisation manipulates them explicitly.

First come the presuppositions: a related pair of terms has a reflexively related
type, and terms related at the universe are exactly types related at the
corresponding code.
\begin{align*}
  \tag{IsType}
  \ValEq{\Gamma}{M}{N}{A}{u}{a}
  &\hspace{.5em}\Rightarrow\hspace{.5em}
  \TVal{\Gamma}{A}{a}
  \label{eq:is-type} \\
  \tag{ToType}
  \ValEq{\Gamma}{M}{N}{A}{u}{\cuniv}
  &\hspace{.5em}\Rightarrow\hspace{.5em}
  \TValEq{\Gamma}{M}{N}{u}
  \label{eq:to-type}
\end{align*}
On the type side, \(\TValEq{\Gamma}{A}{B}{\bot}\) always holds: this is the
degenerate clause of the definition. On the term side, the analogous statement
for a \(\bot\) \emph{term} witness is a genuine property.
\begin{align*}
  \tag{Bot}
  \presup{a \ty_n \cuniv} \qquad
  \presup{\TVal{\Gamma}{A}{a}}
  &\hspace{.5em}\Rightarrow\hspace{.5em}
  \ValEq{\Gamma}{M}{N}{A}{\bot}{a}
  \label{eq:bot}
\end{align*}
Both relations are partial equivalence relations.
\begin{align*}
  % \tag{Left}
  % \ValEq{\Gamma}{M}{N}{A}{u}{a}
  % &\hspace{.5em} \Rightarrow \hspace{.5em}
  % \Val{\Gamma}{M}{A}{u}{a}
  % \label{eq:left} \\
  % \tag{LeftTy}
  % \TValEq{\Gamma}{A}{B}{a}
  % &\hspace{.5em} \Rightarrow \hspace{.5em}
  % \TVal{\Gamma}{A}{a}
  % \label{eq:left-ty} \\
  \tag{Symm}
  \presup{\conv{\Gamma}{M}{N}[A]} \qquad
  \ValEq{\Gamma}{M}{N}{A}{u}{a}
  &\hspace{.5em} \Rightarrow \hspace{.5em}
  \ValEq{\Gamma}{N}{M}{A}{u}{a}
  \label{eq:symm} \\
  \tag{SymmTy}
  \TValEq{\Gamma}{A_1}{A_2}{a}
  &\hspace{.5em} \Rightarrow \hspace{.5em}
  \TValEq{\Gamma}{A_2}{A_1}{a}
  \label{eq:symm-ty} \\
  \tag{Trans}
  \begin{cases}
    \presup{\conv{\Gamma}{M_1}{M_2}[A]} \qquad
    \presup{\conv{\Gamma}{M_2}{M_3}[A]} \\
    \ValEq{\Gamma}{M_1}{M_2}{A}{u}{a} \qquad
    \ValEq{\Gamma}{M_2}{M_3}{A}{u}{a}
  \end{cases}
  &\hspace{.5em} \Rightarrow \hspace{.5em}
  \ValEq{\Gamma}{M_1}{M_3}{A}{u}{a}
  \label{eq:trans} \\
  \tag{TransTy}
  \TValEq{\Gamma}{A_1}{A_2}{a} \qquad
  \TValEq{\Gamma}{A_2}{A_3}{a}
  &\hspace{.5em} \Rightarrow \hspace{.5em}
  \TValEq{\Gamma}{A_1}{A_3}{a}
  \label{eq:trans-ty}
\end{align*}
A heterogeneous transitivity is moreover available when the second step is taken at
the universe code, with an arbitrary type witness \(s\) on the first:
\begin{align*}
  \tag{Trans'}
  \ValEq{\Gamma}{A_1}{A_2}{\univ}{a}{s} \qquad
  \ValEq{\Gamma}{A_2}{A_3}{\univ}{a}{\cuniv}
  &\hspace{.5em} \Rightarrow \hspace{.5em}
  \ValEq{\Gamma}{A_1}{A_3}{\univ}{a}{s}
  \label{eq:trans-prime}
\end{align*}
The term relation transports along related types:
\begin{align*}
  \tag{Conv}
  \TValEq{\Gamma}{A}{B}{a} \qquad
  \ValEq{\Gamma}{M}{N}{A}{u}{a}
  &\hspace{.5em} \Rightarrow \hspace{.5em}
  \ValEq{\Gamma}{M}{N}{B}{u}{a}
  \label{eq:conv}
\end{align*}
Furthermore, the relations are monotone in their semantic witnesses, both by
lowering and by increasing them.
\begin{align*}
  \tag{MonoTypeIncr}
  \begin{cases}
    \presup{u \ty_n a} \quad \presup{u \ty_n a'} \quad \presup{\TVal{\Gamma}{A}{a'}} \\
    a \leq a' \quad \quad \ValEq{\Gamma}{M}{N}{A}{u}{a}
  \end{cases}
  &\hspace{.5em} \Rightarrow \hspace{.5em}
  \ValEq{\Gamma}{M}{N}{A}{u}{a'}
  \label{eq:mono-type-incr}\\
  \tag{MonoTypeDecr}
  a \leq a' \quad \presup{u \ty_n a} \quad \presup{a' \ty_n \cuniv} \quad
  \ValEq{\Gamma}{M}{N}{A}{u}{a'}
  &\hspace{.5em} \Rightarrow \hspace{.5em}
  \ValEq{\Gamma}{M}{N}{A}{u}{a}
  \label{eq:mono-type-decr} \\
  \tag{MonoTypeDecrTy}
  a \leq a' \quad \presup{a \ty_n \cuniv} \quad \presup{a' \ty_n \cuniv} \quad
  \TValEq{\Gamma}{A}{B}{a'}
  &\hspace{.5em} \Rightarrow \hspace{.5em}
  \TValEq{\Gamma}{A}{B}{a}
  \label{eq:mono-type-decr-ty} \\
  \tag{MonoTerm}
  u \leq u' \quad \presup{u \ty_n a} \quad \presup{u' \ty_n a} \quad
  \ValEq{\Gamma}{M}{N}{A}{u'}{a}
  &\hspace{.5em} \Rightarrow \hspace{.5em}
  \ValEq{\Gamma}{M}{N}{A}{u}{a}
  \label{eq:mono-term}
\end{align*}
Finally, the relations are closed under joins of compatible witnesses:
\begin{align*}
  \tag{JoinTy}
  \begin{cases}
    \presup{a_1 \compat a_2} \quad \presup{a_1 \ty_n \cuniv} \quad \presup{a_2
    \ty_n \cuniv} \\
    \TValEq{\Gamma}{A}{B}{a_1} \quad
    \TValEq{\Gamma}{A}{B}{a_2}
  \end{cases}
  &\hspace{.5em} \Rightarrow \hspace{.5em}
  \TValEq{\Gamma}{A}{B}{a_1 \vee a_2}
  \label{eq:join-ty} \\
  \tag{Join}
  \begin{cases}
    \presup{u_1 \compat u_2} \quad \presup{u_1 \ty_n a} \quad \presup{u_2 \ty_n a} \\
    \ValEq{\Gamma}{M}{N}{A}{u_1}{a} \quad
    \ValEq{\Gamma}{M}{N}{A}{u_2}{a}
  \end{cases}
  &\hspace{.5em} \Rightarrow \hspace{.5em}
  \ValEq{\Gamma}{M}{N}{A}{u_1 \vee u_2}{a}
  \label{eq:join}
\end{align*}

A main consequence of the monotonicity and stability laws above is that the
relations do not depend on the stage, as long as the semantic witnesses are
available at it.

\begin{corollary}[Stability]
  For \(u, a \in \comp[n] \subseteq \comp[n+1]\) we have
  \[
    \begin{aligned}
      \TValEq[n]{\Gamma}{A}{A'}{a} &\quad\iff\quad \TValEq[n+1]{\Gamma}{A}{A'}{a} \\
      \ValEq[n]{\Gamma}{M}{M'}{A}{u}{a} &\quad\iff\quad \ValEq[n+1]{\Gamma}{M}{M'}{A}{u}{a}
    \end{aligned}
  \]
\end{corollary}

Consequently the relations are independent of the stage at which \(u\) and
\(a\) are taken: for \(u, a \in \comp\) we may drop the index and simply write
\(\TValEq[]{\Gamma}{A}{A'}{a}\) and \(\ValEq[]{\Gamma}{M}{M'}{A}{u}{a}\).

Finally, the term relation is invariant under typed head reduction of \emph{both} of
its components, at the same semantic witnesses.  Since it is stated as an
equivalence, it yields head expansion as well as head reduction.

\begin{lemma}[Head reduction]
  \label{lem:whr}
  Assume \(\convred{\Gamma}{M}{M'}[A]\), \(\convred{\Gamma}{N}{N'}[A]\) and
  \(\conv{\Gamma}{M}{N}[A]\).  Then
  \[
    \ValEq[]{\Gamma}{M}{N}{A}{u}{a}
    \quad\iff\quad
    \ValEq[]{\Gamma}{M'}{N'}{A}{u}{a}.
  \]
\end{lemma}

The next step is to strengthen the fundamental lemma, making it relative to
two parallel substitutions \(\sigma, \sigma' \colon \Delta \to \Gamma\).
This genuinely requires a general environment
fitting \(\Gamma\): the \(\bot\) environment suffices only for the closed
instances, since interpreting a \(\P\)-type already forces its codomain to be
considered under arbitrary arguments.

We extend the logical relation to substitutions as follows:
\(\sigma, \sigma' \colon \Delta \to \Gamma\) are \emph{related}
over an environment \(\rho\) fitting \(\Gamma\) when, for every \((x : A) \in \Gamma\),
\[
  \ValEq[]{\Delta}{\sigma(x)}{\sigma'(x)}{A\sigma}{\rho}{\sem{A}[\rho]}.
\]

With these tools in hand we can state the fundamental theorem of the logical relation,
called \emph{adequacy}:
\begin{theorem}[Adequacy
  \lean{Adequacy}{450}{LR.adequacy}\,\rocq{finelt/adequacy}{2461}{adequacySub}\,\agda{ID/Adequacy/Value}{199}{adequacyEqSub2}]
  \label[theorem]{thm:adequacy}
  Let \(\rho\) be an environment that fits \(\Gamma\),
  and let \(\sigma, \sigma' \colon \Delta \to \Gamma\) be
  related over \(\rho\).  Then, for \(a \le \sem{A}[\rho]\) with \(a \ty \cuniv\),
  and for \(u \le \sem{M}[\rho]\) with \(u \ty a\),
  \begin{align*}
    \conv{\Gamma}{A}{A'}[\univ] &\implies
    \TValEq[]{\Delta}{\subs{A}{\sigma}}{\subs{A'}{\sigma'}}{a}, \\
    \conv{\Gamma}{M}{M'}[A] &\implies
    \ValEq[]{\Delta}{\subs{M}{\sigma}}{\subs{M'}{\sigma'}}{\subs{A}{\sigma}}{u}{a}.
  \end{align*}
\end{theorem}

As a special case when \(\Delta = \Gamma\),
\(\sigma = \sigma' = \mathrm{id}\) and \(\rho = \bot\), we get the fundamental lemma.
\begin{theorem}[Fundamental lemma]
  \label[theorem]{thm:fundamental}
  For any \(a \le \sem{A}[\bot]\) with \(a \ty \cuniv\),
  and \(u \le \sem{M}[\bot]\) with \(u \ty a\),
  \begin{align*}
    \conv{\Gamma}{A}{A'}[\univ] &\implies \TValEq[]{\Gamma}{A}{A'}{a}, \\
    \conv{\Gamma}{M}{M'}[A] &\implies \ValEq[]{\Gamma}{M}{M'}{A}{u}{a}.
  \end{align*}
\end{theorem}

\subsection{Consequences}
\label{sec:lr-consequences}
From the fundamental lemma, injectivity of Π follows.

\begin{corollary}[Injectivity of Π
  \lean{Adequacy}{2501}{forallE_inv}\,\rocq{finelt/adequacy}{2653}{piInjectivity}\,\agda{ID/PiInjectivity}{209}{piInjectivity}]
  \label[theorem]{thm:pi-injectivity}
  If \(\conv{\Gamma}{\P_{x : A}B}{\P_{x : A'}B'}[\univ]\), then
  \(\conv{\Gamma}{A}{A'}[\univ]\) and \(\conv{\Gamma,x:A}{B}{B'}[\univ]\).
\end{corollary}

\begin{proof}
  By the fundamental lemma, it follows that
  \(\TValEq[]{\Gamma}{\P_{x : A}B}{\P_{x : A'}B'}{\cpi(\bot,\bot)}\), from
  which the definition of the logical relation implies in particular
  \(\conv{\Gamma}{A}{A'}[\univ]\) and \(\conv{\Gamma,x:A}{B}{B'}[\univ]\).
\end{proof}

Injectivity of \(\P\) supplies exactly the inversion necessary to deal with the
case of application in progress and preservation (subject reduction).

\begin{theorem}[Subject reduction
  \lean{Consequences}{657}{WHRed.subject_red}\,\rocq{finelt/adequacy}{2730}{subject_red}\,\agda{ID/SubjectReduction}{126}{subject-red1}]
  \label[theorem]{thm:subject-reduction}
  If \(\typing{\Gamma}{M}[A]\) and \(M \red M'\), then \(\typing{\Gamma}{M'}[A]\).
\end{theorem}
\begin{theorem}[Progress
  \lean{Consequences}{881}{progress}\,\rocq{finelt/adequacy}{2843}{progress}]
  \label{thm:progress}
  If $\typing{{}}{M}[A]$, then either \(M\) is an abstraction, a product, or a universe,
  or there is an \(M'\) with \(M \sred M'\).
\end{theorem}

\section{Extensions}
\label{sec:extensions}

The domain model of \cref{sec:domain-model} is deliberately coarse: it collapses
structure that
the syntax carefully keeps apart. In other words, although equality
is \emph{preserved} (definitionally equal terms are interpreted by equal domain elements),
it is certainly not \emph{reflected}: many terms which are not definitionally equal
have the same semantics.

Far from being a defect, this makes the technique scale better.
Indeed, the logical relation of \cref{sec:logical-relation} relates \emph{syntactic} terms to
\emph{semantic} witnesses, and the witnesses serve only to structure the induction: every
conversion we ultimately read off, and hence every definitional inversion principle we
obtain, is a conversion of the object theory.  A model may therefore identify terms that the
type theory distinguishes, without weakening the conclusions drawn from it.

\paragraph{Universe hierarchies}
The clearest instance of this is the universe hierarchy.  Proof assistants such as
\Agda, \Lean and \Rocq stratify types in a countable tower \(\univ[0] \ty \univ[1]
\ty \cdots\),
whereas \(\comp\) offers a single code \(\cuniv\), with \(\cuniv \ty \cuniv\).
Since \(\univ \ty
\univ\) is so permissive, we may interpret \emph{every} \(\univ[i]\)
by \(\cuniv\): levels are simply erased.  Each rule of the hierarchy is easily validated by
the flattened interpretation: \(\univ[i] \ty \univ[i+1]\) becomes \(\cuniv \ty \cuniv\), the
level-annotated \(\P\) formation rules collapse to the single rule \textsc{Pi}, and
cumulativity, where present, becomes trivial; so soundness (\cref{thm:soundness}) needs
no extra work.

This collapses a lot of information, but is not an issue: we can still run the
construction of \cref{sec:logical-relation}, with the stratified syntax
on the syntactic side and the unstratified domain \(\dom\) on the semantic side.
The clause at a \(\cpi\)-code
records \(\conv{\Gamma}{A}{A'}[\univ[i]]\) and \(\conv{\Gamma,x:A}{B}{B'}[\univ[j]]\) in the
\emph{stratified} theory, so the fundamental lemma yields injectivity of \(\P\)
there directly.
The semantic witness never mentions a level, and never needs to.

Indeed, although the model itself cannot tell \(\univ[i]\) from \(\univ[j]\), the logical
relation still can: its semantic witnesses carry no level, but its syntactic components
can and do speak of them: it suffices to sharpen
the clause at the universe code so that the two types reduce to the \emph{same} universe,
\[
  \TValEq{\Gamma}{A}{A'}{\cuniv} \quad\defeq\quad \exists i,\quad
  \convred{\Gamma}{A}{\univ[i]} \quad\wedge\quad \convred{\Gamma}{A'}{\univ[i]},
\]
to read off the sort injectivity theorem
$\conv{\Gamma}{\univ[i]}{\univ[j]}\Rightarrow i=j$ as a consequence of the fundamental lemma.
This construction works for both cumulative and non-cumulative universe hierarchies.

This is the general shape of the trade-off: the model may forget as much as it likes,
as long as it keeps elements for
\begin{enumerate*}[label=(\arabic*)]
\item injectivity properties we wish to obtain, and
\item constructors that can be eliminated by later operations ($\lambda$, pairs, $0$ and
  $\suc$)
\end{enumerate*}.

Each extension below follows the same pattern: the syntax gains constructors and rules,
the compact elements \(\comp\) gain new codes, and the logical relation gains a
clause at each
new code. We must show additional clauses in \cref{thm:soundness} and \cref{thm:adequacy},
and then we learn that the main results
\cref{thm:pi-injectivity,thm:subject-reduction,thm:progress}
hold for the extended type theory as well.

\subsection{Dependent sum types}

Our first extension adds the dependent sum type \(\Sigma_{x : A}B\), its pairing
constructor and two projections, together with the two
projection β-rules and surjective pairing (the η-rule for pairs):
\begin{mathparpagebreakable}
  \inferdef[Sig]{\typing{\Gamma}{A}[\univ] \\ \typing{\Gamma, x :
  A}{B}[\univ]}
  {\typing{\Gamma}{\Sigma_{x : A}B}[\univ]} \and
  \inferdef[Pair]{\typing{\Gamma}{M}[A] \\ \typing{\Gamma}{N}[\subs{B}{M}]}
  {\typing{\Gamma}{(M,N)}[\Sigma_{x : A}B]} \\
  \inferdef[Proj1]{\typing{\Gamma}{M}[\Sigma_{x : A}B]}
  {\typing{\Gamma}{M.1}[A]} \and
  \inferdef[Proj2]{\typing{\Gamma}{M}[\Sigma_{x : A}B]}
  {\typing{\Gamma}{M.2}[\subs{B}{M.1}]} \\
  \inferdef[Proj1β]{
    \typing{\Gamma}{M}[A] \\
    \typing{\Gamma}{N}[\subs{B}{M}]
  }{
    \conv{\Gamma}{(M,N).1}{M}[A]
  } \and
  \inferdef[Proj2β]{
    \typing{\Gamma}{M}[A] \\
    \typing{\Gamma}{N}[\subs{B}{M}]
  }{
    \conv{\Gamma}{(M,N).2}{N}[\subs{B}{M}]
  } \and
  \inferdef[η]{\typing{\Gamma}{M}[\Sigma_{x : A}B]}
  {\conv{\Gamma}{(M.1, M.2)}{M}[\Sigma_{x : A}B]}
\end{mathparpagebreakable}
As with the rules for \(\P\), the presuppositions \(\typing{\Gamma}{A}[\univ]\) and
\(\typing{\Gamma, x : A}{B}[\univ]\) are kept implicit; the congruence rules
for \(\Sigma\), pairing and the projections are the evident ones.

\paragraph{Head reduction}
% On top of the function \(\beta\)-step
% \((\l x : A.M)\,N \sred \subs{M}{N}\) and its congruence in function position,
\(\Sigma\) contributes the two projection steps and their congruences:
\begin{mathpar}
  (M,N).1 \sred M \and
  (M,N).2 \sred N \and
  \inferdef{M \sred M'}{M.1 \sred M'.1} \and
  \inferdef{M \sred M'}{M.2 \sred M'.2}
\end{mathpar}
% Its reflexive--transitive closure is written \(\red\); the \emph{typed} head
% reduction \(\convred{\Gamma}{M}{N}[A]\) is the restriction of \(\red\) to
% well-typed terms (\(M \red N\) with \(\typing{\Gamma}{M}[A]\)), whose generating
% steps are the computation rules \textsc{Proj1β}, \textsc{Proj2β} read left to
% right; by subject reduction it entails \(\conv{\Gamma}{M}{N}[A]\).

\paragraph{Domain model} The finite elements \(\comp[n]\) gain a code
\(\mathsf{Σ}(b,c)\) for the sum (like \(\cpi\)) together with pairs \((u,v)\),
subject to \((\bot,\bot) = \bot\). Like for \(\cabs\), this equation ensures that
we can interpret η. We set \(\mathsf{Σ}(b,c) \ty_{n+1} \cuniv\), and
\((u,v) \ty_{n+1} \mathsf{Σ}(b,c)\) whenever \(u \ty_n b\) and \(v \ty_n c(u)\).

\paragraph{Logical relation} At the code \(\mathsf{Σ}(b,c)\) the two relations
follow the pattern of \(\cpi\), using the projections in place of application: two
types are related when they both reduce to sums with related domains and
codomains. For the type itself we can literally reuse the same definition at for $\Pi$-types.
\begin{align*}
  \TValEq[n+1]{\Gamma}{T}{T'}{\mathsf{Σ}(b,c)} \quad\defeq\quad
  \begin{cases}
    \convred{\Gamma}{T}{\Sigma_{x : B}C}[\univ] \qquad
    \convred{\Gamma}{T'}{\Sigma_{x : B'}C'}[\univ] \\
    \TValEq[n+1]{\Gamma}{\P_{x:B}C}{\P_{x:B'}C'}{\cpi(b,c)}
  \end{cases}
\end{align*}
Two terms are related at a sum type through their two projections:
\begin{align*}
  \ValEq[n+1]{\Gamma}{M}{M'}{T}{(u,v)}{\mathsf{Σ}(b,c)} \quad\defeq\quad
  \begin{cases}
    \presup{\TVal[n+1]{\Gamma}{T}{\mathsf{Σ}(b,c)}}\\
    \ValEq{\Gamma}{M.1}{M'.1}{B}{u}{b} \\
    \ValEq{\Gamma}{M.2}{M'.2}{\subs{C}{M.1}}{v}{c(u)}
  \end{cases}
\end{align*}
Since \((\bot,\bot) = \bot\), the \(\bot\) clause stays degenerate.

The rest of the adequacy proof poses no particular difficulty.

\subsection{Unit type}
\label{sec:unit}

The nullary analogue of \(\Sigma\) is a unit type \(\unit\), with a single element
\(\star\) and a definitional η-rule identifying every inhabitant with it:
\begin{mathpar}
  \inferdef[Unit]{\ctxty{\Gamma}}{\typing{\Gamma}{\unit}[\univ]} \and
  \inferdef[UnitI]{\ctxty{\Gamma}}{\typing{\Gamma}{\star}[\unit]} \and
  \inferdef[Unitη]{\typing{\Gamma}{M}[\unit]}{\conv{\Gamma}{M}{\star}[\unit]}
\end{mathpar}
The η-rule \textsc{Unitη} collapses \(\unit\) to a single point up to
conversion, and in particular any two of its inhabitants are convertible, so \(\unit\)
is a (strict) proposition.

\paragraph{Domain model} We add a single code \(\cunit \ty_{n+1} \cuniv\) whose
only element is \(\bot\); thus \(\star\) is interpreted by \(\bot\) and \(\unit\) is
the one-point domain. This is the same as adding a code for \(\star\) but then requiring
that is equal to \(\bot\), just as we did to handle η for functions and pairs.

\paragraph{Logical relation} We need only define
\(\TValEq[n+1]{\Gamma}{A}{B}{\cunit}\defeq
  \begin{cases}
    \convred{\Gamma}{A}{\unit}[\univ] \\
    \convred{\Gamma}{B}{\unit}[\univ].
\end{cases}\)

\subsection{General recursion}

Since we are already dealing with non-termination, we can embrace it and add a
fixed point combinator \(\operatorname{Y}\) for general recursion, making the system
explicitly non-normalising (notwithstanding $\univ\ty\univ$ which yields
much more complicated non-normalizing terms). With the addition of natural
numbers in \cref{sec:nat}, we obtain a type theory which is a sort of
dependently typed PCF \cite{Plotkin1977}.
The combinator is typed at any type, and unfolds one step:
\begin{mathpar}
  \inferdef[Y]{\typing{\Gamma}{A}[\univ] \\
  \typing{\Gamma}{F}[A \rightarrow A]}
  {\typing{\Gamma}{\fix F}[A]} \and
  \inferdef[Yβ]{\typing{\Gamma}{A}[\univ] \\
  \typing{\Gamma}{F}[A \rightarrow A]}
  {\conv{\Gamma}{\fix F}{F\,(\fix F)}[A]}
\end{mathpar}

\paragraph{Head reduction} The new step unfolds the fixed point: \(\fix F \sred
F\,(\fix F)\).

\paragraph{Domain model} General recursion adds no new elements: \(\fix F\) is
interpreted as the supremum \(\bigvee_{k} \sem{F}^{k}(\bot)\) of its finite approximations.
Consequently, it contributes no logical relation clause either --~\(\fix F\)
is reached through the supremum of these approximations.

\subsection{Natural numbers}
\label{sec:nat}

To demonstrate our ability to deal with large elimination, we can add a type
\(\mathbb{N}\) of natural numbers, its two constructors and dependent
pattern-matching.%
\footnote{A recursor can be derived by working in a syntax with the \(\fix\)
  from the previous section. Alternatively, we can use the same idea as for
  universes at the beginning of the section, keep the non-terminating \(\fix\)
in the semantics only and have a syntax with only a primitive recursor.}
We
write \(\suc M\) for the successor and, for the (large) eliminator,
\(\case{M}{N}{P}\) where \(N\) is the branch for \(0\) and \(P\) the branch for
the successor (a function of the predecessor):
\begin{mathparpagebreakable}
  \inferdef[Nat]{\ctxty{\Gamma}}{\typing{\Gamma}{\mathbb{N}}[\univ]} \and
  \inferdef[Zero]{\ctxty{\Gamma}}{\typing{\Gamma}{0}[\mathbb{N}]} \and
  \inferdef[Succ]{\typing{\Gamma}{M}[\mathbb{N}]}
  {\typing{\Gamma}{\suc M}[\mathbb{N}]} \and
  \inferdef[Case]{
    \typing{\Gamma, x : \mathbb{N}}{C}[\univ] \\
    \typing{\Gamma}{M}[\mathbb{N}] \\\\
    \typing{\Gamma}{N}[\subs{C}{0}] \\
    \typing{\Gamma}{P}[\P_{n : \mathbb{N}}\subs{C}{\suc n}]
  }{
    \typing{\Gamma}{\case[C]{M}{N}{P}}[\subs{C}{M}]
  } \\
  \inferdef[Case-0]{
    \typing{\Gamma, x : \mathbb{N}}{C}[\univ] \\\\
    \typing{\Gamma}{N}[\subs{C}{0}] \\
    \typing{\Gamma}{P}[\P_{n : \mathbb{N}}\subs{C}{\suc n}]
  }{
    \conv{\Gamma}{\case[C]{0}{N}{P}}{N}[\subs{C}{0}]
  } \and
  \inferdef[Case-S]{
    \typing{\Gamma, x : \mathbb{N}}{C}[\univ] \\
    \typing{\Gamma}{M}[\mathbb{N}] \\\\
    \typing{\Gamma}{N}[\subs{C}{0}] \\
    \typing{\Gamma}{P}[\P_{n : \mathbb{N}}\subs{C}{\suc n}]
  }{
    \conv{\Gamma}{\case[C]{(\suc M)}{N}{P}}{P\,M}[\subs{C}{\suc M}]
  }
\end{mathparpagebreakable}
% The non-dependent case-splitting is the special case where the motive \(C\) does not
% depend on \(x\), together with the evident congruence rules for \(\suc\)
% and \(\operatorname{case}\).

\paragraph{Head reduction} The new steps fire pattern-matching,
with a congruence in the scrutinee:
\begin{mathpar}
  \begin{aligned}
    \case[C]{0}{N}{P} &\sred N \\[-1em]
    \case[C]{(\suc M)}{N}{P} &\sred P\,M
  \end{aligned} \and
  \inferdef{M \sred M'}{\case[C]{M}{N}{P} \sred \case[C]{M'}{N}{P}}
\end{mathpar}
% As before \(\red\) is the reflexive--transitive closure and
% \(\convred{\Gamma}{M}{N}[A]\) its typed restriction, generated by
% \textsc{Case-0} and \textsc{Case-S}; each yields \(\conv{\Gamma}{M}{N}[A]\) by
% subject reduction.

\paragraph{Domain model} We add a code \(\mathsf{N} \ty_{n+1} \cuniv\) and value
elements \(\mathsf{0}\) and \(\mathsf{S}\,u\). This time, \(\mathsf{S}\,\bot\) is
a proper element, distinct from \(\bot\): observing that a given number is
a successor is a valid observation. We then set \(\mathsf{0} \ty_{n+1} \mathsf{N}\) and
\(\mathsf{S}\,u \ty_{n+1} \mathsf{N}\) when \(u \ty_n \mathsf{N}\).

\paragraph{Logical relation} The code \(\mathsf{N}\) is a base type, so the type
clause only records the reduction and the term clauses read off the canonical
form named by the value code:
\begin{align*}
  \TValEq[n+1]{\Gamma}{A}{A'}{\mathsf{N}} &\defeq
  \convred{\Gamma}{A}{\mathbb{N}}[\univ] \quad\wedge\quad
  \convred{\Gamma}{A'}{\mathbb{N}}[\univ] \\
  \ValEq[n+1]{\Gamma}{M}{N}{A}{\mathsf{0}}{\mathsf{N}} &\defeq
  \begin{cases}
    \presup{\TVal[n+1]{\Gamma}{A}{\mathsf{N}}} \\
    \convred{\Gamma}{M}{0}[\mathbb{N}] \qquad
    \convred{\Gamma}{N}{0}[\mathbb{N}]
  \end{cases} \\
  \ValEq[n+1]{\Gamma}{M}{N}{A}{\suc u}{\mathsf{N}} &\defeq
  \begin{cases}
    \presup{\TVal[n+1]{\Gamma}{A}{\mathsf{N}}} \\
    \convred{\Gamma}{M}{\suc M'}[\mathbb{N}] \qquad
    \convred{\Gamma}{N}{\suc N'}[\mathbb{N}] \\
    \ValEq{\Gamma}{M'}{N'}{\mathbb{N}}{u}{\mathsf{N}}
  \end{cases}
\end{align*}
Thus a value of code \(\mathsf{0}\) forces both \(M\) and \(M'\) to reduce to \(0\), and one
of code \(\suc u\) forces both to reduce to successors whose predecessors are
related at \(u\).

\subsection{Proof-relevant identity type}

Continuing with inductive types, we tackle Martin-Löf's proof-relevant
identity type \(\Id_{A}\,M\,N\), its reflexivity
constructor \(\refl_{M}\), and transport \(\tr\).%
\footnote{The same development goes through for the fully dependent eliminator
  \(\J\), whose motive ranges over both endpoints and the proof, and this is what
  the \Agda{} formalisation does.  Since the domain model and the logical relation
  depend only on the code \(\Id(a,u,v)\) and the element \(\crefl{w}\),
  and not on the eliminator, the same model works for both versions;
  only syntactic typing and reduction rules change, as well as the proof of the
  logical relation.
We present transport, as in the \Lean{} development, for clarity and brevity.}
The motive \(C\) is a
family over the carrier alone, \(\typing{\Gamma, x : A}{C}[\univ]\), and
transport carries an inhabitant of \(\subs{C}{M}\) along a proof of
\(\Id_{A}\,M\,N\) to an inhabitant of \(\subs{C}{N}\):
\begin{mathpar}
  \inferdef[Id]{
    \typing{\Gamma}{A}[\univ] \\ \typing{\Gamma}{M}[A] \\
    \typing{\Gamma}{N}[A]
  }{
    \typing{\Gamma}{\Id_{A} M\,N}[\univ]
  } \and
  \inferdef[Ref]{
    \typing{\Gamma}{A}[\univ] \\ \typing{\Gamma}{M}[A]
  }{
    \typing{\Gamma}{\refl_{M}}[\Id_{A} M\,M]
  } \\
  \inferdef[Tr]{
    \typing{\Gamma}{A}[\univ] \\
    \typing{\Gamma}{M}[A] \\ \typing{\Gamma}{N}[A] \\
    \typing{\Gamma, x : A}{C}[\univ] \\
    \typing{\Gamma}{X}[\subs{C}{M}] \\
    \typing{\Gamma}{H}[\Id_{A}\,M\,N]
  }{
    \typing{\Gamma}{\transp{C}{X}{H}[\Id_{A}\,M\,N]}[\subs{C}{N}]
  } \and
  \inferdef[Trβ]{
    \typing{\Gamma}{A}[\univ] \\
    \typing{\Gamma}{M}[A] \\
    \typing{\Gamma, x : A}{C}[\univ] \\
    \typing{\Gamma}{X}[\subs{C}{M}]
  }{
    \conv{\Gamma}{\transp{C}{X}{\mathsf{refl}\,M}[\Id_{A}\,M\,M]}{X}[\subs{C}{M}]
  }
\end{mathpar}
Rule \textsc{Trβ} fires only on the reflexive diagonal --~both
endpoints are \(M\), and the proof is literally \(\mathsf{refl}\,M\)~-- so the
redex \(\transp{C}{X}{\mathsf{refl}\,M}[\Id_{A}\,M\,M]\) and the contractum \(X\)
already share the type \(\subs{C}{M}\), and the rule needs no endpoint- or
motive-equality side
conditions. We leave the carrier and the two endpoints
implicit in the notation \(\transp{C}{X}{H}\); the mechanisation handles them
explicitly.

\paragraph{Head reduction} The identity type adds the transport step on a
reflexive proof, with a congruence in the eliminated proof:
\begin{mathpar}
  \transp{C}{X}{\mathsf{refl}\,P}[\Id_{A}\,M\,N] \sred X \and
  \inferdef{H \sred H'}{\transp{C}{X}{H}[\Id_{A}\,M\,N] \sred
  \transp{C}{X}{H'}[\Id_{A}\,M\,N]}
\end{mathpar}
Unlike \textsc{Trβ}, this untyped step fires on \emph{any} reflexivity proof
\(\mathsf{refl}\,P\), whatever its witness: nothing forces \(M,N,P\) to be equal.
Its typed restriction is generated by \textsc{Trβ}, as before.

\paragraph{Domain model} We add a code \(\mathsf{Id}(a,u,v)\) for the type and
\(\crefl{w}\) for reflexivity, with \(\mathsf{Id}(a,u,v) \ty_{n+1} \cuniv\)
when \(u,v \ty_n a\) and
\(\crefl{w} \ty_{n+1} \Id(a,u,v)\) when additionally \(w \ty_n a\) and
\(w \le u\), \(w \le v\). Thus, a valid witness of reflexivity is a common lower bound
of both sides.
To interpret the eliminator, we define
$u\le\sem{\transp{C}{X}{H}[\Id_{A}\,M\,N]}[\rho]$ whenever either $u=\bot$ or $u\le
u'\le\sem{X}[\rho]$, $\crefl{v}\le\sem{H}[\rho]$,
$v\le\sem{M}[\rho]\cap\sem{N}[\rho]$, $v\ty\sem{A}[\rho]$, and
$u'\ty\sem{C}[\rho, x \mapsto v]$.

\paragraph{Logical relation} At \(\Id(a,u,v)\) the type clause relates the
carrier and both endpoints, and two reflexive proofs at witness \(\crefl{w}\)
relate \(P\) and \(P'\) when both reduce to reflexivity proofs whose witnesses are
related at \(w\) and, by the reflexive equation, equal to both endpoints:
\begin{align*}
  &\TValEq[n+1]{\Gamma}{A}{A'}{\Id(a,u,v)} \quad\defeq\\
  &\hspace{1em}
  \begin{cases}
    \convred{\Gamma}{A}{\Id_{B}\,M\,N}[\univ] \qquad
    \convred{\Gamma}{A'}{\Id_{B'}\,M'\,N'}[\univ] \\
    \presup{\conv{\Gamma}{B}{B'}[\univ]} \qquad
    \presup{\conv{\Gamma}{M}{M'}[B]} \qquad
    \presup{\conv{\Gamma}{N}{N'}[B]} \\
    \TValEq{\Gamma}{B}{B'}{a} \qquad
    \ValEq{\Gamma}{M}{M'}{B}{u}{a} \qquad
    \ValEq{\Gamma}{N}{N'}{B}{v}{a}
  \end{cases}
\end{align*}
\begin{align*}
  &\ValEq[n+1]{\Gamma}{H}{H'}{A}{\crefl{w}}{\Id(a,u,v)} \quad\defeq\\
  &\hspace{1em}
  \begin{cases}
    \convred{\Gamma}{A}{\Id_{B}\,M\,N}[\univ] \qquad
    \convred{\Gamma}{H}{\mathsf{refl}\,P}[\Id_{B}\,M\,N] \qquad
    \convred{\Gamma}{H'}{\mathsf{refl}\,P'}[\Id_{B}\,M\,N] \\
    \presup{w \ty_n a} \qquad w \le u \qquad w \le v \\
    \presup{\conv{\Gamma}{P}{M}[B]} \qquad
    \presup{\conv{\Gamma}{P}{N}[B]} \qquad
    \presup{\conv{\Gamma}{P}{P'}[B]} \\
    \ValEq{\Gamma}{P}{M}{B}{w}{a} \qquad
    \ValEq{\Gamma}{P}{N}{B}{w}{a} \qquad
    \ValEq{\Gamma}{P}{P'}{B}{w}{a}
  \end{cases}
\end{align*}
The logical relation forces the endpoints to be convertible whenever a proof
of \(\Id_{B}\,M\,N\) reduces to a reflexivity proof, which is what is needed to show the \textsc{Tr} case of \cref{thm:adequacy}.

Far from identifying proofs, the relation \emph{records} each proof's witness and
demands that the witnesses agree.  The model accordingly keeps enough elements
apart to see that the identity type is genuinely proof-relevant. We can use this to
prove UIP is not derivable:

\begin{theorem}[Failure of definitional UIP \lean{Consequences}{948}{not_UIP}]
  \label{thm:not-uip}
  The following rule of definitional uniqueness of identity proofs (equivalently, the
  \(\mathsf{K}\) rule) is not derivable in \MLTTo{}:
  \[
    \inferrule{\typing{\Gamma}{H}[\Id_{A}\,M\,M]}
    {\conv{\Gamma}{H}{\mathsf{refl}\,M}[\Id_{A}\,M\,M]}
  \]
\end{theorem}

\begin{proof}
  Take \(\Gamma \coloneq (h \ty \Id_{\univ}\,\univ\,\univ)\).  Were the rule admissible,
  we would have
  \(\conv{\Gamma}{h}{\mathsf{refl}\,\univ}[\Id_{\univ}\,\univ\,\univ]\).
  Consider the environment \(\rho \coloneq (h \mapsto \crefl{\bot})\).  It fits
  \(\Gamma\), since \(\crefl{\bot} \ty \Id(\cuniv,\cuniv,\cuniv)\): indeed
  \(\bot \ty \cuniv\), and \(\bot \le \cuniv\).  By \cref{thm:soundness} we would
  then have \(\sem{h}[\rho] = \sem{\mathsf{refl}\,\univ}[\rho]\).  But \(\sem{h}[\rho] =
  \{u\mid u\le \crefl{\bot}\}\), whereas
  \(\sem{\mathsf{refl}\,\univ}[\rho]=\{u\mid u\le \crefl{\cuniv}\}\), and
  \(\crefl{\cuniv} \not\le \crefl{\bot}\), a contradiction.
\end{proof}

The argument has the same shape as the one for disjointness of type constructors:
exhibit a fitting environment that separates two terms, then read soundness
contrapositively.  It is also the most that this model can deliver.  It can never
refute the \emph{inhabitation} of a type, because \(\bot\) inhabits every semantic
type by rule \textsc{Bot} of \cref{fig:sem-typing}; and indeed \emph{propositional}
UIP is not refutable at all here, since under \(\univ \ty \univ\) every type is
inhabited.

There is nothing special about natural numbers or Martin-L\"of identity type,
and our technique should extend to deal with general indexed inductive types
\cite{CoquandPaulin1990}.

\subsection{A universe of strict propositions}

Our last extension adds a universe \(\prop\) of
\emph{strict} (definitionally proof-irrelevant) propositions, lying below
\(\univ\) and closed under dependent products.  A proposition is a type,
and any two inhabitants of a proposition are definitionally equal:
\begin{mathpar}
  \inferdef[Prop]{\ctxty{\Gamma}}{\typing{\Gamma}{\prop}[\univ]} \and
  \inferdef[Prop-U]{\typing{\Gamma}{A}[\prop]}{\typing{\Gamma}{A}[\univ]} \and
  \inferdef[Fun-Prop]{\typing{\Gamma}{A}[\univ] \\
  \typing{\Gamma, x : A}{B}[\prop]}
  {\typing{\Gamma}{\P_{x : A}B}[\prop]} \\
  \inferdef[Irr]{
    \typing{\Gamma}{A}[\prop] \\
    \typing{\Gamma}{M}[A] \\
    \typing{\Gamma}{N}[A]
  }{
    \conv{\Gamma}{M}{N}[A]
  } \and
  \inferdef[Prop-Uβ]{\conv{\Gamma}{M}{N}[\prop]}{\conv{\Gamma}{M}{N}[\univ]}
\end{mathpar}
Rule \textsc{Prop-U} makes \(\prop\) a sub-universe of \(\univ\), and
\textsc{Fun-Prop} closes \(\prop\) under \(\P\) with codomain in
\(\prop\).  Rule \textsc{Irr} is definitional proof irrelevance: it collapses
every strict proposition to at most one element up to conversion, and
\textsc{Prop-Uβ} propagates the resulting equalities from \(\prop\) to
\(\univ\). There is no need to amend
reduction: proof irrelevance is only an extensionality rule.

\paragraph{Domain model} The only new compact element is a code
\(\mathsf{Prop} \ty_{n+1} \cuniv\), lying below \(\cuniv\). Proof irrelevance
adds no elements: instead, it imposes a constraint on the model, wherein
we must be able to prove that if \(u \ty_n a \ty_n \mathsf{Prop}\) then \(u=\bot\):
every proposition, when interpreted as a type, is the singleton \(\{\bot\}\).
For Π-types, this holds because any inhabitant of the Π-type must
be of the form $\cabs(f)$, but for each $x$, $f(x)\ty_{n-1}b(x)\ty_{n-1}\prop$
implies $f(x)=\bot$ recursively, so $\cabs(f)=\cabs(\bot)=\bot$. We can
introduce other type constructors, such as $\top\ty\prop$ with rules like
\cref{sec:unit}, provided we maintain this property that types in $\prop$ have no
non-trivial elements.

\section{Comments About the Mechanised Proofs}
\label{sec:formalisation}

We have three mechanised proofs of the results of the previous section, completed in
\Agda, \Lean, and \Rocq.%
\footnote{Note to reviewers: the formalisations are available as anonymous supplementary
  material for this submission.
  It includes a comparison document that lists the specific definitions in each
system including pointers to where they may be found.}
The former provided the basis for initial investigations, while the
latter two are candidates for future inclusion in the \leanf~\cite{Carneiro2024} and
\MetaRocq~\cite{MetaCoq2025} projects. Together, the \Agda and \Lean versions
cover all the features described in \cref{sec:extensions}.
For convenience, the \Rocq mechanisation uses
the axioms of dependent functional extensionality, propositional extensionality,
and proof irrelevance, which are also assumed by \Lean's standard library.

\paragraph{Use of AI}
All proofs were constructed with the assistance of \textsf{Claude AI}
(\textsf{Opus 4.7} for the \Lean mechanisation, and \textsf{Opus 4.8} for the other
two), which assisted with the translation of proof structure between systems.
Specifically, the \Agda formalisation was entirely written by \textsf{Claude},
with a lot of human guidance. The other two formalisations were human-written,
with \textsf{Claude} assistance, particularly to translate existing proofs
from other systems. Nonetheless, as soon as the proofs became more subtle
and less routine, the assistance offered by \textsf{Claude} became more erratic.

\paragraph{Syntactic metatheory}
Each formalisation declares the typing and conversion rules of \MLTTo{}
(\cref{fig:syntax}) together with a deterministic head-reduction relation.  The
conversion rules are not stated in their most economical form: they carry as
extra premises the presuppositions of their conclusion.  The congruence rules for
\(\P\) and \(\l\), for instance, require the codomain and the body equations to
hold under \emph{both} domains \(A\) and \(A'\), and the β- and
η-rules require both sides of the equation to be well-typed at the common
type.  These premises are redundant: the resulting judgement is proved to derive
exactly the same equations as the economical one.%
\footnote{\texttt{IsDefEq.iff} in the \Lean{} development.}
They are there for proof ergonomics: every argument by induction on a
conversion derivation (and in particular \cref{thm:adequacy}) then finds, in each
case, the induction hypotheses it needs
already at hand, instead of having to appeal to a separately established
presupposition lemma.

The \Rocq and \Agda versions use an intrinsically scoped representation of
terms, while \Lean enforces scoping only with the typing judgement. The \Lean version
furthermore defines the typing relation as the diagonal of the conversion
relation, while the \Agda and \Rocq ones maintain separate, mutually defined
judgements.

\textsf{AutoSubst}~\cite{Stark2019} provided support for substitution lemmas in the
\Rocq version.

\paragraph{Compact elements}
In \Agda and \Rocq, we define \(\comp\) directly as a nested inductive datatype,
using a list of pairs to represent step functions, and carve out
\(\comp[n]\). In \Lean, the definition of \(\comp[n]\)
instead goes by induction on \(n\), and \(\comp\) is built from \(\Sigma_{n :
\nat} \comp[n]\).
Each version also defines
a predicate to identify well-formed terms: in functions, compatible inputs
must produce compatible outputs and $\lambda$ terms must not be $\bot$.  The
formalisations differ in the treatment of these predicates: in \Agda and \Rocq,
they are passed explicitly, whereas in \Lean they are packaged
as a subset type. All formalisations work with an unquotiented representation,
instead partially enforcing a canonical representation (\ie the body of an
abstraction cannot represent \(\bot\)), and implicitly handling \(\comp\) as a setoid.%
\footnote{Whose equivalence relation is the one associated with the pre-order
on compact elements, making it easier to handle.}

The proofs of properties about the ordering and typing relations proceed by
strong induction on a measure of the compact elements.
The three formalisations stratify differently: \Rocq keeps a
single type of compact elements, with the rank as an external measure; \Agda
indexes membership by a stage, recovered at the canonical stage determined by the
ranks of the subject and its type; \Lean makes the type of compact elements itself
level-indexed. Semantic typing (\cref{fig:sem-typing}) is correspondingly an
inductive relation in \Rocq, with the ranks implicit, but a recursive definition
in \Agda (on the stage) and in \Lean, structural on the shape and Boolean-valued.
All three state the \textsc{Lam} and \textsc{Pi} rules in the intensional form of
\cref{thm:fun-typing}, quantifying over the entries of the function's
representation (it is frequently useful to know that the condition is
a finite conjunction) and recover the extensional form of \cref{fig:sem-typing}
as a derived equivalence.%
\footnote{For instance \texttt{WShape.HasDom.iff} and
\texttt{WShape.HasTypeLam.iff} in the \Lean{} development.}
% All three versions satisfy the same closure properties.

\paragraph{Model and logical relation}
\Agda and \Rocq define the domain model as a predicate of compact elements
(\cref{fig:sem-compacts}) by recursion on syntactic terms, while in \Lean, this
definition is an (indexed) inductive relation. All three formalisations prove
the soundness theorem by induction on the typing and conversion relations
(\cref{thm:soundness}).

The three formalisations define the
logical relations by recursion on the rank/stage of the compact elements and
their semantic types. As the definitions presuppose the typing relation for
these elements, the \Rocq version takes this derivation as an additional
argument.  As described in \cref{sec:lr-properties}, the formalisations must
first prove properties about the logical relation by recursion on rank/stage,
in several mutual groups. The fundamental lemma itself is proven by induction
on the syntactic typing/conversion judgement.

\section{Conclusion}
\label{sec:conclusion}

We have proposed a new proof technique for the meta-theory of dependent type
systems. We use ideas from domain theory, in order to establish definitional
inversion principles, important meta-theoretic properties with far-reaching
consequences. Moreover, we demonstrate that this is a robust approach, in
that it adapts with only minor variations to many extensions, including
some, such as η for unit or a fixed point combinator, which are difficult to deal
with using traditional methods. We believe this new approach opens up many avenues.

The most evident application is the specification of real-life proof assistants in
themselves,
\eg the \leanf and \MetaRocq projects. While \MetaRocq already comes with a
substantial body of meta-theory, its reliance on confluence and untyped conversion
has until now been a big blocker to reach feature parity with \Rocq, or at least
the ability to check large libraries with its verified kernel. \leanf, on the
other hand, has focused on feature parity from the start, meaning that until
now we lacked good meta-theoretic backing to verify the algorithm. Hopefully
our proof technique will help unblock both projects and open up a new era of
fully verified kernels for realistic proof assistants. However, although our
example extensions cover many of the theoretical difficulties one could expect to
face when studying the meta-theory of \Lean or \Rocq, we still are far from
having covered all the nitty-gritty details of these systems, around \eg
(co-)inductive types or proof irrelevance with axiom $\mathsf{K}$, and scaling our
approach remains
a significant challenge.

On the dependently typed programming front, the possibility of showing progress
and preservation of non-terminating type theories seems enticing. On the one hand,
having an alternative to rewriting theory might make it less frightening to extend
the definitional equalities of these languages with more interesting equations,
starting with η. Having a way to argue about the well-behaviour of non-normalising
dependent type systems might also bring more attention to features which have
been regarded with suspicion because they break termination, such as
gradual dependent types \cite{Eremondi2019,LennonBertrand2022,Eremondi2022}, where
there is a sharp trade-off between normalisation and the gradual guarantees.
It will be interesting to see what other fresh takes on dependently typed
programming will be enabled by these new meta-theoretic possibilities.

On the more theoretical front, much still remains to explore to better understand
our model and adequacy proof. One obvious topic of enquiry is the relationship
to step-indexed logical relations. With domain theory, this is the other natural
candidate to handle non-termination and solve fixed point equations. Can we
use it to obtain similar results as ours? Are there different trade-offs, and
what are they? Step-indexing comes with a very well-developed body of
specialised logics, in particular those developed around the Iris \cite{Jung2018}
project. Would this be of any help in proofs like the ones we do? Could we
use synthetic domain theory \cite{Hyland1991} to similarly reason at a
higher abstraction level than with our elementary proofs? Or could we
re-cast our technique as an instance of a gluing model \cite{Sterling2021a,Bocquet2023},
which would shed more light on its relation with normalisation proofs?

In \cref{sec:logical-relation} we briefly touched on the question of
definitional inversion for neutrals, which we currently do not cover.
\textcite{Coquand2018a} show such properties by using a more involved
kind of model, where neutrals are not degenerate. Another avenue we find
promising is to revisit confluence proofs, but with added type information.
This is usually not possible, since one needs definitional inversion for types
for typed confluence to work well, but this is precisely what confluence is
supposed to establish. However, with these established by other means, it
seems possible to rely on confluence to show other properties, such
as the missing inversion for neutrals. See \textcite{Carneiro2019} for
work going in this direction.

Last but not least, the model comes equipped with intersection and singleton types:
every domain element \(u\) defines a type \(\{u\} \coloneq \{v \mid v \le u\}\), and
we can form the intersection of two types \(p\) and \(q\) (qua finitary projectors)
by considering the finitary projector obtained from those compact elements
that are fixed points of both \(p\) and \(q\). These types have a strong connection
with linear logic, and it would be interesting to see whether this sheds any
light on the problem of designing good linear dependent type systems.

\printbibliography

\end{document}